\documentclass[11pt,draft]{article}
\usepackage{fullpage}
\usepackage[final]{pdfpages}
\usepackage{amsmath, amssymb, amsthm, amsfonts, latexsym, color, paralist, cancel,ifdraft,mathtools}
\usepackage[final]{hyperref}
\usepackage{graphicx}
\usepackage{circuitikz}
\newcommand{\reg}[1]{{\color{gray}#1}}
\usepackage{authblk}
\usepackage{cleveref}
\usepackage{float}
\usepackage{tikz}
\usetikzlibrary{positioning, arrows,calc, fit, shapes, chains, matrix}
\usetikzlibrary{decorations.pathreplacing}
\usepackage{tikzsymbols}
\usepackage{ifthen}
\usepackage{subcaption}
\usepackage{adjustbox}
\usepackage{qcircuit}
\newcommand{\initial}{\textsf{initial}}
\newcommand{\unary}{\textsf{unary}}
\newcommand{\setunsat}{\textsf{set-unsat}}
\newcommand{\calC}{\mathcal{C}}
\def\01{\{0,1\}}
\newcommand{\stepsrandomwalk}{\frac{16d^2}{\delta^2}\ln\frac{2d|V|}{\delta}} 
\newcommand{\concat}{\,||\,}

\ifdraft{}{}
\usepackage[refpage]{nomencl}
\makenomenclature

\hypersetup{
  colorlinks   = true, %
  urlcolor     = blue, %
  linkcolor    = {blue!55!black}, %
  citecolor   = red %
}

\newcommand{\NOT}{\ensuremath{\mathrm{NOT}\xspace}} %
\newcommand{\CNOT}{\ensuremath{\mathrm{CNOT}\xspace}} %
\newcommand{\CCNOT}{\ensuremath{\mathrm{CCNOT}\xspace}} %
\newcommand{\OR}{\ensuremath{\mathrm{OR}\xspace}} %
\newcommand{\AND}{\ensuremath{\mathrm{AND}\xspace}} %
\newcommand{\problem}[1]{\ensuremath{\mathsf{#1}\xspace}} %

  \theoremstyle{plain}
  \newtheorem{theorem}{Theorem}[section]
  \newtheorem{lemma}[theorem]{Lemma}
  \newtheorem{corollary}[theorem]{Corollary}

  \theoremstyle{definition}
  \newtheorem{definition}[theorem]{Definition}
    \newtheorem{remark}[theorem]{Remark}

\newcommand{\comment}[1]{}

\newcommand{\defeq}{\stackrel{\mathrm{def}}{=}}
\newcommand{\poly}[1]{\mathrm{poly( #1 )}}

\newcommand{\eps}{\varepsilon}

\newcommand{\class}[1]{\textup{#1}}
\newcommand{\QMA}{\class{QMA}}   
\newcommand{\MA}{\class{MA}}

\newcommand{\NP}{\class{NP}}   
\newcommand{\PSPACE}{\class{PSPACE}}   
\newcommand{\PPSPACE}{\class{PPSPACE}}

\newcommand{\SetCSP}{\problem{SetCSP}}
\newcommand{\NameSuccCC}{Approximately-clean approximate-connected-component}
\newcommand{\nameSuccCC}{approximately-clean approximate-connected-component}
\newcommand{\NameBaseSuccCC}{Clean connected component}
\newcommand{\nameBaseSuccCC}{clean connected component}
\newcommand{\CleanCC}{\ensuremath{\problem{CCC}}}
\newcommand{\FixedCleanCC}{\ensuremath{\problem{FCCC}}}
\newcommand{\ACleanCC}{\ensuremath{\problem{ACAC}}}
\newcommand{\ACleanCCeps}{\ensuremath{\problem{ACAC}_{\eps}}}

\newcommand{\kSetCSP}{$k$\text{-}\problem{SetCSP}}
\newcommand{\CSP}{\problem{CSP}}
\newcommand{\kCSP}{$k$\textsf{-}\problem{CSP}}

\newcommand\defclass[5]{
\begin{definition}[#1]\label{#2}
#3 \\
\;\;\textbf{Completeness.} #4 \\
\quad\quad\textbf{Soundness.} #5
\end{definition} 
}

\newcommand\defproblem[5]{
\begin{definition}[#1]\label{#2}
#3 \\
\hspace{1cm} \textbf{Yes.} #4 \\
\noindent\;\;\textbf{No.} #5
\end{definition} 
}

 \newcommand{\ayes}{A_{\rm yes}} 
  \newcommand{\ano}{A_{\rm no}} 
\usepackage{authblk}
  
\begin{document}
\pagenumbering{gobble}
\title{Two combinatorial MA-complete problems}

\author[1]{Dorit Aharonov}
\author[2]{Alex B. Grilo}
\affil[1]{Hebrew University of Jerusalem, Jerusalem, Israel}
\affil[2]{Sorbonne Universit\'{e}, CNRS, LIP6, France}

\date{}
\maketitle
\begin{abstract}
    Despite the interest in the complexity class \MA{}, the randomized analog of
  \NP{}, there are just a few known natural (promise-)\MA{}-complete
  problems. The first such problem was found by Bravyi and Terhal (SIAM Journal of Computing
  2009); this result was then followed by Crosson, Bacon and Brown (PRE 2010) and then by Bravyi (Quantum Information and Computation 2015).
Surprisingly, each of these problems is either from or inspired by quantum computation.
This fact makes it hard for 
   {\em classical}  complexity
  theorists to study these problems, and prevents potential progress, e.g., on the important question of derandomizing MA. 
  
 In this note we define two new natural combinatorial
  problems and we prove their \MA{}-completeness.
  The first problem, that we call \nameSuccCC{} ($\ACleanCC$), gets as input a succinctly described graph, some of whose vertices are marked. The problem is to decide whether there is a connected component whose vertices are {\it all unmarked}, or the graph is {\em far} from having this property.
The second problem, called \SetCSP{}, generalizes in a novel way 
the standard constraint satisfaction problem (\CSP{}) into constraints involving {\it sets} of strings.

  Technically, our proof that \SetCSP{} is \MA{}-complete is a fleshing out of an observation made in (Aharonov and Grilo, FOCS 2019), where it was noted that a restricted case of Bravyi and Terhal's MA complete problem (namely, 
  the {\it uniform} case) is 
  already \MA{} complete; and, moreover, that  
  this restricted case can be stated using classical, combinatorial language. 
  The fact that the first, arguably more natural, problem of $\ACleanCC$ 
  is \MA{}-hard follows quite naturally from this proof as well; 
  while containment of $\ACleanCC$ in \MA{} is simple, based on the theory of random walks. 
 
 We notice that this work, along with a translation of  the main 
  result of Aharonov and Grilo to the \SetCSP{} problem, 
  implies that finding a gap-amplification procedure for 
  \SetCSP{} (in the spirit of 
  the gap-amplification procedure introduced in Dinur's PCP proof) 
  would imply \MA{}=\NP{}. In fact, the problem of finding  gap-amplification for \SetCSP{} is 
  {\it equivalent} to the \MA{}=\NP{} problem. This provides an alternative new path towards the major problem of derandomizing \MA{}. 
    Deriving a similar statement regarding gap amplification of a natural restriction of \ACleanCC{} remains an open 
  question.
\end{abstract}
\clearpage
\pagenumbering{arabic}

\section{Introduction}
The complexity class \MA{} is a natural extension of \NP{} proof systems to the
probabilistic setting~\cite{Babai85}. There is a lot of evidence towards the fact that these
two complexity classes are
equal~\cite{ImpagliazzoKW02,Goldreich2011,HastadILL99,BabaiFNW93,NisanW94,ImpagliazzoW97,SudanTV01,KabanetsI2004}, however the proof remains elusive. 
It is even open to show that every problem in \MA{} can be solved in non-deterministic {\it sub-exponential} time. 

Surprisingly, the very first {\em natural} \MA{}-complete\footnote{
For PromiseMA, it is folklore that one can define complete problems by extending NP-complete problems (see, e.g. \cite{WilliansShor19}):
we define an exponential family of 3SAT formulas (given as input succinctly) and 
we have to decide if there is an
assignment that satisfies all of the formulas, or for every assignment, a random
formula in the family will not be satisfied with good probability.} problem, found by 
Bravyi and Terhal~\cite{BravyiT09} only close to $25$ years after the definition of the class (!) is defined using {\it quantum} terminology. 
But why would randomized NP have anything to do with quantum? Bravyi and Terhal show that deciding
if a given stoquastic $k$-Local Hamiltonian\footnote{Stoquastic Hamiltonians sit
between classical Hamiltonians (\CSP{}s) and general quantum Hamiltonians.} is frustration-free or at least inverse
polynomially frustrated, is promise-\MA{}-complete. 
Bravyi~\cite{Bravyi14}
 also proved \MA{}-completeness of yet another quantum Hamiltonian problem. 
A third MA-complete problem was proposed by Crosson, Bacon and Brown \cite{CrossonBB10}, inspired by quantum adiabatic computation.\footnote{Crosson et. al.'s problem asks about the properties of the Gibbs distribution corresponding to the temperature T of a specific family of {\it classical}, rather than quantum, physical systems, defined using local classical Hamiltonians. 
Though inspired by adiabatic quantum computation, this problem 
is in fact defined using {\it classical} terminology, since
the classical Hamiltonians can be viewed as constraint satisfaction systems.
Yet, we note that its definition uses a layer of physical notation, involving Gibbs distribution and temperature. Moreover, when stated using classical terminology, the input to this problem is restricted, in a fairly contrived way, to sets of classical constraints which can be associated with a (noisy) deterministic circuit. Both of these aspects seem to make handling this problem using standard combinatorial tools difficult or at least not very natural.}

This leaves us in a strange situation in which the known \MA{} complete problems are not stated using natural or standard complexity theory terminology. This makes us wonder if there is a fundamental reason why we cannot find classical \MA{}-complete problems, or it is just an ``accident'' that the first \MA{}-complete problems come from quantum computing.  
 Moreover, we think that new natural (classically defined) complete problems for the class \MA{} might enable access to the major open problem of derandomization of \MA{},
 and possibly to other related problems such as 
 PCP for \MA{} \cite{AharonovG19}\footnote{We notice that Drucker~\cite{Drucker11} proves a PCP theorem for \class{AM}; in the definition he uses for \class{AM}, the coins are public and the prover sees them (See Section $2.3$ in \cite{Drucker11}); but his result does not hold when the coins are private, namely for \class{MA}.}. In particular, though the problems proposed in~\cite{BravyiT09,Bravyi14} are very natural within quantum complexity theory, the fact that 
 they are defined within the area 
 of quantum computation seems to pose a language and conceptual barrier that 
 might delay progress on them, and make it hard for  
 {\em classical} complexity theorists to study them. 
 
 The goal of this work is to provide 
classically, combinatorially-defined 
 complete problems for \MA{}. 
 We hope that these definitions lead to 
further understanding of the 
class \MA{} and the \MA{} vs. \NP{} question.

One of these problems, \SetCSP{}, is morally based on 
 the \MA{} complete problem of \cite{BravyiT09}, while the other problem, \ACleanCC, seems to be a rather natural problem on graphs. 
 
  The definition of \SetCSP{} as well as the proof of its \MA{}-completeness rely on a simple but crucial insight: we can translate ``testing history states", a notion familiar in quantum complexity theory, into constraints on {\it sets} of strings. This idea is used here throughout, but in fact can be used  to translate also other quantum results related to stoquastic Hamiltonians, to the classical combinatorial language of constraints on {\it sets of strings} (see \Cref{sec:intuition} as well as \Cref{fig:comparison} for more intuition). However though this translation idea can be viewed as standing at the heart of many of the definitions and results presented here, it is in fact  {\it not} necessary to understand the proofs.

Based on this idea, as well as a simple observation made in \cite{AharonovG19} which says that the problem of \cite{BravyiT09} remains promise-\MA{} complete even when restricted to what is referred 
 to in \cite{AharonovG19} as {\it uniform}
 stoquastic Hamiltonians, we can prove the \MA{}-hardness proof of the \SetCSP{} problem as a  translation of the \MA-hardness proof given in \cite{BravyiT09} into a classical language. We notice that the proof of \cite{BravyiT09} on its own heavily relies on the Quantum Cook-Levin proof of Kitaev \cite{KitaevSV02}.  

It turns out that there is an easy reduction from the problem  
\SetCSP{} to that of \ACleanCC{}, and therefore
the proof that  \ACleanCC{} is \MA{}-hard follows
immediately. 
 
Interestingly, the proof of containment in \MA{} is rather simple for \ACleanCC{} (it is an easy application of a known result from random walk theory.) Finally, using the same reduction, we arrive at a proof that the \SetCSP problem is also in \MA{}\footnote{This in fact gives a significantly simpler version than\cite{BravyiBT06} of the proof of containment in \MA{}, in the restricted case of 
uniform stoquastic Hamiltonians.}. 

We stress that we present all proofs here
avoiding any quantum notation or quantum jargon whatsoever.
The main contribution of this work is in the definitions themselves, initiated by the small but important conceptual idea of the translation mentioned above; this translation thus provides two new, combinatorial,  and natural \MA{}-complete problems, which we believe are amenable to research in 
 a language familiar to (classical) computer scientists.

 We now describe the problems.
  In the \ACleanCC ~problem, we consider an exponentially large graph, accompanied with a function on the vertices that marks some of them. Both graph and the function on the vertices are given implicitly (and succinctly) by a polynomial size circuit. We then ask if there exists a connected component of the graph that is ``clean" (meaning that all of its vertices are {\it unmarked}) or if the graph is $\eps$-far from having this property. The notion of ``far'' is defined as follows: every set of vertices which is close to being a connected component (i.e. its expansion  is smaller than  $\eps$) must have at least an $\eps$-fraction of its vertices marked. In other words, either 
  a set is $\eps$-far from a connected component (i.e. has large expansion) or at least $\eps$ fraction of its vertices are marked. We call this problem {\it\nameSuccCC{}} ($\ACleanCC_\eps$).

 Our second \MA{}-complete problem, called
the {\it Set-Constraint Satisfaction Problem}, or \SetCSP{}, is a somewhat unexpected 
generalization of the standard Constraint Satisfaction Problem (\CSP{}). While a constraint in \CSP{} acts on a single string (deciding if it is valid or not), the generalized constraints act on {\it sets} of strings. 
We call the generalized constraints {\it set-constraints} 
(see \Cref{sec:set-csp} for the exact definition of a set-constraint.). The input to the \SetCSP{} problem is a collection of such set-constraints, and the output is whether there is a set of strings $S$ that satisfies {\it all}  set-constraints in this collection, or any set of strings $S$ is $\eps$-far from satisfying this set-constraint collection (see \Cref{def:satisfiability,def:setunsat} for formal definitions of ``satisfying a set-constraint'' and ``far''). We denote this problem by $\SetCSP_\eps$. 

Following the ideas of \cite{BravyiT09,KitaevSV02}, we show the following three claims: $i)$ 
for every inverse polynomial $\eps$, we have that $\ACleanCCeps$ is in \MA{} (Corollary \ref{cor:containment}); $ii)$ there exists an inverse polynomial function $\eps=\eps(n)$ such that $\SetCSP{}_{\eps}$ is \MA{}-hard (\Cref{cor:ma-hardness-setcsp}); and $iii)$ for all functions $\eps=\eps(n)>0$, there is a polynomial-time reduction from $\SetCSP{}_{\eps}$ to   $\ACleanCC_{\eps/2}$. %
Together, these facts imply the following results (which are proven in \Cref{sec:reduction}).

\begin{theorem}\label{thm:main}
There exists some inverse polynomial $p(x) = \Theta(1/x^3)$ such that for every inverse polynomial $p' < p$, the problems $\SetCSP_{p'(m)}$ and $\ACleanCC_{p'(m)}$ are $\MA$-complete, where $m$ is the size of the $\SetCSP$ or $\ACleanCC$ instance.
\end{theorem}

We stress that the definitions of both these problems and the proofs presented in this paper only use standard combinatorial concepts from set- and graph-theory.

\subsection{Conclusions, future work and open questions}
 In a similar way to what is done in this note, we claim that it is possible to translate several other results from the stoquastic local Hamiltonian language to the \SetCSP{} language, bringing us to very interesting conclusions. 

First, one could strengthen the \MA{}-hardness of \SetCSP{} (\Cref{cor:ma-hardness-setcsp}), whose proof is a translation of the proof of the quantum Cook-Levin theorem \cite{KitaevSV02}, to the restricted case in which the constraints are set on a 2D lattice. 
To do this, one could consider the result of \cite{AharonovDKLLR08} where they prove the $\QMA$-hardness of local Hamiltonians on a 2D lattice (considering only the restricted family of stoquastic Hamiltonians), and translate it to the \SetCSP{} language. 
In this way, it can be proven that 
$\SetCSP{}_\eps$ with inverse polynomial $\eps$, and with set-constraints arranged on a $2D$-lattice with constant size alphabet is still \MA{}-hard.\footnote{We conjecture that this hardness result works even with binary alphabet and we leave such a statement for future work.} We omit here the details of the proof, since it is a straight forward translation of  \cite{AharonovDKLLR08}, similarly to what's done in \Cref{cor:ma-hardness-setcsp}. This leads to the following statement.

\begin{lemma}
There exists some inverse polynomial $p$ such that for every inverse polynomial $p' < p$,
$\SetCSP{}_{p'}$ is 
\MA{}-hard even when each bit participates in $O(1)$ set-constraints, and each set-constraint acts on $O(1)$ bits.
\end{lemma}

We also claim that the main result of \cite{AharonovG19}, which states that some problem called stoquastic local Hamiltonian with constant gap is in \NP{},  can be translated to \SetCSP{} language, using again the same translation. 
This means that the {\it gapped} version of the $\SetCSP{}$ problem 
is in \NP{}.  By the gapped version of the problem,
we mean $\SetCSP{}_\eps$  when $\eps$ is a constant; and where we also require that 
the locality $k$ (number of bits in a set-constraint) and degree $d$ (number of set-constraints each bit participates in) are bounded from above by a constant. More concretely, we have the following.

\begin{lemma}
For any constant $\eps$ and constants $k$ and $d$, 
$\SetCSP{}_{\eps}$  is in 
\NP{} if  each bit participates in $d$ set-constraints, and each set-constraint acts on $k$ bits.
\end{lemma}

Together these two results lead to a very 
important and surprising 
equivalence\footnote{This equivalence was highlighted in \cite{AharonovG19} in a quantum language of stoquastic Hamiltonians}: 

\medskip

\noindent{\bf Proposition.} \MA{}=\NP{} iff 
there is a gap amplification reduction\footnote{In 
the same sense as Dinur's gap amplification reduction for 
CSP\cite{Dinur07}} for \SetCSP{}.  The existence of a gap amplification reduction means that there exists a constant $\eps>0$, such that for every inverse polynomial $p$, there is a polynomial time reduction from $\SetCSP_{p(n)}$ to $\SetCSP_{\eps}$.

\medskip

We note that it is easy to see that \MA{}=\NP{} implies such gap-amplification for \SetCSP{}: if $\MA=\NP$, then we can reduce $\SetCSP_{1/poly}$, which is in $\MA{}$-complete by \Cref{thm:main}, to $\CSP_{O(1)}$, which is \NP{}-complete by the PCP theorem~\cite{Dinur07}; then, since every $\CSP$ instance is also a $\SetCSP$ instance with the same parameters, we have the gap-amplification.  It is the other direction of deriving \MA{}=\NP{} from gap-amplification for \SetCSP, that is the new contribution. This implication suggests a new path to the long standing 
open problem of derandomizing \MA{}.  

Finally, we leave as an open problem deriving such a natural statement regarding gap amplification for the  \ACleanCC{} problem. Though the two problems are technically very related, defining a natural restricted version of the gapped \ACleanCC problem, so that the \cite{AharonovG19} result would apply to show containment in \NP{} (similarly to \SetCSP{} with constant $\eps$, locality $k$ and degree $d$) remains open. The problem is that the locality notions don't have a very natural analogues in the graph language of \ACleanCC. 
\\~\\
\noindent\textbf{Organization.}
We provide notation and a few basic notions in
\Cref{sec:prelim}. In \Cref{sec:succccc}, we define the \nameSuccCC{} problem and prove its containment in \MA{}. The definition of \SetCSP{} and the proof of its \MA{}-hardness are in 
\Cref{sec:set-csp}. The reduction from \SetCSP{} to \ACleanCC{} appear in
\Cref{sec:reduction}. %
In \Cref{ap:classes,ap:rev,ap:graph} we provide some basic background on computational complexity and random walks, and prove some standard technical results we need. In \Cref{app:pspace}, we also prove the \PSPACE{}-completeness of the exact version of the \ACleanCC{} problem.

\section{Preliminaries}
\label{sec:prelim}

For $n \in \mathbb{N}^+$, we denote $[n] = \{0,...,n-1\}$. For any $n$-bit
string, we index its bits from $0$ to $n-1$.
For $x \in \01^n$ and $J \subseteq [n]$, we denote $x|_{J}$ as the substring
of $x$ on the positions contained in $J$.
For $x \in \01^{|J|}$, $y = \01^{n - |J|}$ and $J \subseteq [n]$,
we define
$x^{\reg{J}}y^{\reg{\overline{J}}}$ to be
the unique $n$-bit string $w$ such that $w|_{J} = x$ and $w|_{\overline{J}} = y$, where
 $\overline{J} = [n] \setminus J$. For
two strings, $x$ and $y$, we denote by $x \concat y$ their concatenation, and $|x|$ denotes the number of bits in $x$.  

\subsection{Complexity classes}
A (promise) problem  $A = (A_{yes}, A_{no})$ consists of two non-intersecting sets 
 $A_{yes}, A_{no} \subseteq \{0,1\}^*$. 
For completeness, we add in \Cref{ap:classes} the definitions of the two main complexity classes that are considered in this work: \NP{} and \MA{}.  

The standard definition of
\MA{} \cite{Babai85} requires yes-instances to be accepted with probability at least $\frac{2}{3}$, but it has been shown that there is no change in the computational power  if we
require the verification algorithm to always accept yes-instances~\cite{Zachos1987,Goldreich2011}.

\subsection{Reversible circuits}
\label{sec:reversible}
It is folklore that the verification algorithms for
\NP{} and \MA{} can be converted into a uniform family of 
polynomial-size Boolean circuits,
made of reversible gates, 
 $\{\NOT, \CNOT, 
\CCNOT\}$ by making use of additional auxiliary bits
initialized to $0$, with only linear overhead. 
For randomized circuits, we can also assume the circuit uses only reversible gates, by assuming
that the random bits are part of input. See Appendix \ref{ap:rev} for more details. 

Let $G \in \{\NOT, \CNOT, \CCNOT\}$ be a gate to be applied on the set of bits $J$ out of some
$n$-bit input $x$. 
 We slightly abuse notation (but make it much shorter!) 
 and denote 
$G(x) \defeq x|_{\overline{J}}^{\reg{\overline{J}}}G(z|_{J})^{\reg{J}}$.
Namely, we understand the action of the $k$-bit gate $G$ on an $n>k$ bit string $x$ by applying $G$ only on the relevant bits, and leaving all other bits intact.

\section{\NameSuccCC{} problem}
\label{sec:succccc}
In this section, our goal is to present the  \nameSuccCC{} problem and prove its containment in \MA{}.  But before that, we explain the {\em exact} version of this problem.

For a fixed parameter $n$, we consider a graph $G$ of $2^n$ nodes, which is described by a classical circuit $C_G$ of size (number of gates) $\poly{n}$, as follows. For simplicity, we represent each vertex of $G$ as an $n$-bit string, and $C_G$, on input $x \in \01^n$, outputs the (polynomially-many) neighbors of $x$ in $G$.\footnote{We notice that usually we succinctly describe graphs by considering a circuit that, on input $(x,y)$, outputs $1$ iff $x$ is connected to $y$. In our result it is crucial that given $x$, we are able to efficiently compute {\it all} of its neighbors.} We are also given a circuit $C_M$, which when given input $x \in \01^n$, outputs a bit indicating whether the vertex $x$ is marked or not. 

We define the \nameBaseSuccCC{} problem ($\CleanCC{}$) which, on input $(C_G,C_M)$, asks if $G$ has a connected component where {\em all vertices are  unmarked} or if all connected components have at least one marked element. We give now the formal definition of this problem.

\defproblem{\NameBaseSuccCC (\CleanCC{})}
{def:succ-ccc}{
Fix some parameter $n$. An instance of the \nameBaseSuccCC{} problem consists of two classical $\poly{n}$-size circuits $C_G$ and $C_M$. $C_G$ consists of a circuit succinctly representing a graph with $G = (V,E)$ where $V = \01^n$ and each vertex has degree at most $\poly{n}$. On input $x \in \01^n$, $C_G$ outputs all of the neighbors of $x$. $C_M$ is a circuit that on input $x \in \01^n$, outputs a bit. The problem then consists of distinguishing the two following cases:
}
{There exists one non empty connected component of $G$ such that $C_M$ outputs $0$ on all of its vertices.}
{In every connected component of $G$, there is at least one vertex for which $C_M$ outputs $1$.}

We show in \Cref{app:pspace} that this problem is \PSPACE{}-complete. Our focus here is to study the {\em approximate} version of \CleanCC{}, where we ask whether $G$ has a clean connected component or it is ``{\it far}'' from having this property, meaning that for every set of vertices $S$, the boundary\footnote{
As defined in \Cref{ap:graph}, the boundary of a set of vertices $S \subseteq V$ is defined as $\partial_{G}(S) = \{\{u,v\} \in E: u \in S, v \not\in S\}$.
} of $S$ is at least $\eps|S|$ (and therefore $S$ is far from being a connected component), or, if the boundary is small (i.e. $S$ is close to a connected component) it contains at least $\eps|S|$ marked elements. Here is the definition of the problem.

\begin{definition}[\NameSuccCC ($\ACleanCCeps$)]
Same as \Cref{def:succ-ccc} with the following difference for no-instances:

\noindent \textbf{No.}
For all non empty $S \subseteq V$, we have that either
\[|\partial_{G}(S)| \geq \eps|S| \text{ \quad or \quad }
\left|\left\{x \in S : C_M(x) = 1 \right\}\right| \geq \eps |S|.
\]
\end{definition}

In the remainder of this section, we show that $\ACleanCCeps$ is in \MA{} for every inverse polynomial $\eps$. 

\subsection{Inclusion in \MA}\label{S:ma}

The idea of the proof is as follows. The prover sends a vertex that belongs to the (supposed) clean connected component and then the verifier performs a random-walk on $G$ for sufficiently (but still polynomially-many) steps and rejects if the random walk encounters a marked element.

In order to prove that this verification algorithm is correct, we first prove a technical lemma regarding the random-walk on no-instances.

\begin{lemma}[In NO-instances the random walk reaches marked nodes quickly]
  \label{lem:soundness-ma}
 Let $\eps$ be an inverse polynomial function of $n$. If $(C_G,C_M)$ is a no-instance of $\ACleanCCeps$,
  then there exist 
  polynomials $q_1$ and $q_2$ such that  for every $x \in \01^n$,
  a $q_1(n)$-step lazy random walk\footnote{As is defined in \Cref{ap:graph}, in a lazy random walk, at every step we stay in the current vertex with probability $\frac{1}{2}$, and with probability $\frac{1}{2}$ we choose a random neighbor of the current vertex uniformly at random and move to it.} starting at $x$ reaches a marked vertex with
  probability at least
  $\frac{1}{q_2(n)}$.
\end{lemma}
\begin{proof} 
Let $x$ be some initial (unmarked) vertex (notice that we can assume that $x$ is unmarked since otherwise 
the lazy random walk reaches a marked vertex with probability $1$). Let $G_x$ be the
  connected component of $x$ in $G$ and $V_x$ be the set of vertices in the connected
  component of $x$. We partition $V_x$ into $A$, the unmarked vertices in $V_x$ 
  and $B =
  V_x\setminus A$, the marked vertices in $V_x$.

  We want to upper bound the size of the edge boundary of any set $S \subseteq A$ which contains only unmarked strings.
We 
  claim that by the conditions of the lemma, it must be 
  that for any $S \subseteq A$, 
  we have the following bound:  
  \begin{equation}
  |\partial_{G}(S)| \geq \eps|S|,
  \end{equation} 
  otherwise $S$ would contradict the fact that $(C_G,C_M)$ is a no-instance, since it only contains unmarked vertices. 

  We can now ask how fast does a lazy random walk 
  starting from $x$, reach an element in $B$. 
  This is a well known question from random walk theory, 
  and it can be stated as follows.

\begin{lemma}[Escaping time of high conductance subset]\label{lem:hitting-time}
  Let $G = (V =A\cup B,E)$ be a simple (no multiple edges) 
  undirected connected graph, such that for
  every $v \in A$
  $d_G(v) \leq d$ , and such that for some $\delta < \frac{1}{2}$, 
  for all $A' \subseteq A$, we have that 
  $|\partial_G(A')| \geq \delta |A'|$.
Then  a $\left( 
   \stepsrandomwalk
  \right)$-step lazy random walk starting in any $v \in A$ reaches some vertex $u
  \in B$ with probability at least $\frac{\delta}{4d}$.
\end{lemma}
  We defer the proof of this lemma to \Cref{ap:graph} and we apply it using $G = G_x$, $A$, $B$, 
  $d$ is the ($\poly{n}$-large) maximum degree of $G$,
  and $\delta \geq \eps$.
  It follows that a $\left(\frac{16d^2}{\eps^2} \left(n + \ln\left(2d/\eps\right)\right)\right)$-step lazy random walk starting on
  any $x \in G_x$ reaches a marked vertex with probability at least
  $\frac{\eps}{4d}$.
\end{proof}

From the previous lemma, we can easily achieve the following.

\begin{corollary}\label{cor:containment}
For any inverse polynomial $\eps$, $\ACleanCCeps$ is in \MA.
\end{corollary}
\begin{proof}
The witness for the $\ACleanCC$ instance consists of some
string $x$, which is supposed to be in a clean connected component of $G$. Define $q_1$ and $q_2$ as the same polynomials of
  \Cref{lem:soundness-ma}, namely, let $q_1=\left(\frac{16d^2}{\eps^2} \left(n + \ln\left(2d/\eps\right)\right)\right)$ and $q_2 = \frac{4d}{\eps}$. The $\MA{}$-verification algorithm consists of repeating $nq_2(n)$ times the following process: 
  start from $x$, perform a $q_1(n)$-step lazy random walk in $G$,  reject if any of these walks
  encounters a marked vertex, otherwise accept. 
  
  If $(C_G,C_M)$ is a yes-instance, then the prover can 
  provide a vertex $x$ which belongs to the clean connected component.  
  Any walk
  starting from $x$ remains in its connected component and thus it will never encounter a marked vertex and the verifier accepts with probability $1$.

  If $(C_G,C_M)$ is a no-instance, from \Cref{lem:soundness-ma}, each one of the
  random walks finds a marked vertex with probability at least $\frac{1}{q_2(n)}$. Thus, if we perform $nq_2(n)$ lazy random walks, the probability
  that at least one of them finds a marked vertex is exponentially close to $1$.
\end{proof}

\section{Set-Constraint Satisfaction Problem}
\label{sec:set-csp}
In this section we present the Set
Constraint Satisfaction Problem (\SetCSP{}), and then prove its \MA{}-hardness. 

\subsection{Definition of the \SetCSP{} problem} 

\medskip

We start by recalling the standard Constraint Satisfaction Problem (\CSP{}). We
choose to present \CSP{} in a way which is more adapted to our generalization (but still
equivalent to the standard definition). An instance of \kCSP{} is a sequence of
constraints $C_1,...,C_m$. In this paper, we see each constraint $C_i$ as a tuple $(J(C_i), Y(C_i))$, where $J(C_i)\subseteq [n]$ is a subset of at most $k$ distinct 
elements of $[n]$ (these are referred to as ``the bits on which the constraint acts") and $Y(C_i)$ is a subset of $k$-bit strings, namely $Y(C_i) \subseteq \01^{|J(C_i)|}$ (these are called the ``allowed strings"). 
We say that a string $x \in \01^n$ {\it satisfies} $C_i$ if $x|_{J(C_i)} \in Y(C_i)$.
The familiar problem of \kCSP, in this notation, is to decide whether there exists an $n$-bit string $x$ which satisfies $C_i$, for all $1 \leq i \leq m$.

\medskip

In \kSetCSP{}, our generalization of \kCSP{}, 
the constraints $C_i$ are replaced by what we call ``set-constraints'' which are satisfied by (or alternatively,  that {\it allow}), {\em sets} of strings $S \subseteq \{0,1\}^n$.

\begin{definition}[Set constraint] 
A $k$-local set-constraint $C$ consists of a) 
a tuple of $k$ distinct elements of $[n]$, denoted 
$J(C)$ (we have $|J(C)|=k$, and we refer to $J(C)$ as the bits which the 
set-constraint $C$ acts on), 
and b) 
a collection of sets of strings, $Y(C) = \{Y_1,...,Y_{\ell}\}$, where 
$Y_i \subseteq \{0,1\}^k$ is a set of $k$-bit strings and $Y_i \cap Y_j = \emptyset$ for all distinct $1 \leq i,j \leq \ell$. 
\end{definition} 

\begin{definition}[A set of strings satisfying a constraint]
  \label{def:satisfiability}
We say that a set of strings $S\subseteq \{0,1\}^n$ {\it satisfies} the $k$-local set-constraint $C$
if first, the $k$-bit restriction of any string in $S$ to $J(C)$ is
contained in one of the sets in $Y(C)$, i.e., for all $x \in S$, 
$x|_{J(C)} \in \bigcup_j Y_j$. 
Secondly
we require that if
$x\in S$ and $x|_{J(C)} \in Y_j$, 
then for every $y$ such that $y|_{J(C)} \in Y_j$ and $y|_{\overline{J(C)}} =
x|_{\overline{J(C)}}$, then $y \in S$. In other words,  for any string $s\in S$, one can replace its $k$-bit restriction to the bits $J(C)$, which is a string in 
some $Y_j\in Y(C)$, with a different $k$-bit string in $Y_j$, and the resulting string $s'$ must also be in $S$.  
\end{definition} 

An instance of $\kSetCSP$ consists of $m$ such $k$-local
set-constraints, and we ask if there is some non-empty $S \subseteq \01^n$ 
that satisfies each of the set constraints, or if any set $S$ of $n$-bit strings
is {\em far} from satisfying the collection of set-constraints. 
How to define {\it far}? 
We quantify the distance from satisfaction using a generalization of the familiar notation 
of $\textsf{unsat}$ from PCP theory \cite{Dinur07}; we denote the generalized notion by $\setunsat(\calC,S)$. 
Intuitively, this quantity captures how much we need to modify $S$ in order to satisfy 
the collection of set-constraints. 
Making this definition more precise 
will require some work; However we believe it already makes some sense intuitively, 
so we present the 
definition of the \SetCSP{} problem 
now, and then provide the exact definition of $\setunsat(\calC,S)$ in  \Cref{sec:frustration}. 

\defproblem{$k$-local Set Constraint Satisfaction problem
($\kSetCSP_\eps$)}
{def:set-csp}{Fix 
the two constants $d,k \in \mathbb{N}^+$, as well as a monotone function $\eps : \mathbb{N}^+
\rightarrow (0,1)$ to be some parameters of the problem. An instance to the {\em
$k$-local Set Constraint Satisfaction problem} is a sequence of $m(n)$ $k$-local
set-constraints
  $\calC = (C_1,...,C_m)$ on $\01^n$, 
  where $m$ is some polynomial in 
  $n$.
  Under the promise 
  that one of the following two holds, decide whether: 
}
{There exists a non-empty $S \subseteq \01^n$ that satisfies all set constraints in $\calC$: $\setunsat(\calC,S)=0$}
{For all $S \subseteq \01^n$, 
$\setunsat(\calC,S)\ge \eps(n)$.}

\subsection{Satisfiability, frustration and the definition of $\setunsat(\calC,S)$.} 
\label{sec:frustration}
We present some concepts required for 
the formal definition of $\setunsat(\calC,S)$.
\begin{definition}[$C$-Neighboring strings]
  \label{def:neighbors}
Let $C$ be some set-constraint.  
Two distinct strings $x$ and $y$ are
  said to be $C$-neighbors if 
$x|_{\overline{J(C)}} = y|_{\overline{J(C)}}$  and
  $x|_{J(C)},y|_{J(C)} \in Y_i$, for some $Y_i \in Y(C)$.
  We call a string $x$ a $C$-neighbor 
  of $S$ if there 
  exists a string $y\in S$ such that $x$ is a $C$-neighbor of $y$.
\end{definition}

We also define the $C$-longing strings in a set of strings $S$: these are the strings that are in $S$ but are $C$-neighbors of some string that is not in $S$.

\begin{definition}[$C$-Longing strings]\label{def:missing}
Given some set $S\subseteq\{0,1\}^n$ and a set-constraint $C$, $x \in  S$ is a $C$-longing\footnote{The term "$C$-longing" reflects the sentiment that the string $x$ "wants" to be together with $y$ in $S$; the set constraint makes sure that there is an energy penalty if this is not the case.} string with respect to $S$ if $x$ is a $C$-neighbor of some $y \not\in S$.
\end{definition}

A useful definition is that of {\it bad} strings for some set-constraint $C$, which in short are the strings that do not appear in any subset of $Y(C)$. 

\begin{definition}[$C$-Bad string]\label{def:bad}
Given a set-constraint $C$, with $Y(C)=\{Y_1,...Y_\ell\}$, 
a string $x \in \{0,1\}^n$ is $C$-bad if $x|_{J(C)} \not\in \bigcup_i Y_i$.  We abuse the notation and whenever $x$ is $C$-bad, we say $x \not\in Y(C)$. 
\end{definition}
The following complementary definition will be useful: 
\begin{definition}[Good string]\label{def:good}
We say that a string is $C$-good if it is not $C$-bad.
We say that it is a ``good string for the set-constraint collection $\calC$'' if it is $C$-good for all set-constraints $C\in \calC$. 
When the collection $\calC$ is clear from context (as it is throughout this note), we omit mentioning of the set-constraints collection $\calC$ and just say that the string
is ``good". 
\end{definition} 

Given \Cref{def:bad,def:missing} above, it is easy to see that 
a set of strings $S\subseteq \{0,1\}^n$ satisfies 
a set-constraint $C$ (by \Cref{def:satisfiability}) iff $S$ contains no $C$-bad strings and no $C$-longing strings. 

We now provide a way to quantify how far $S$ is from satisfying $C$.

\begin{definition}[Satisfiability of set-constraints]\label{def:setunsat}
Let  $S \subseteq \01^n$ and $C$ be a $k$-local set-constraint. Let $B_C$ be the
  set of $C$-bad strings in $S$ and $L_C$  be the set of $C$-longing strings in $S$. Note that $L_C \cap B_C=\emptyset$. The
  $\setunsat$ value (which we sometimes refer to as the {\it frustration}) of a set-constraint $C$ with respect to $S$, is defined by 
\begin{align}\label{unsat}
\setunsat(C,S) = \frac{|B_C|}{|S|} + \frac{|L_C|}{|S|}
\end{align}

\noindent Given a collection of $m$ $k$-local set-constraints $\calC = (C_1,...,C_m)$, its
  $\setunsat$ value (or frustration) with respect to $S$ is defined as average frustration of the different set-constraints: 
\begin{align}\label{unsatave}
\setunsat(\calC ,S) = \frac{1}{m}\sum_{i=1}^m \setunsat(C_i,S).
\end{align}
We also define the frustration of the set collection $\calC$: 
\begin{align}\label{unsatmin}
\setunsat(\calC) = \min_{\substack{S \subseteq \01^n \\  S \ne \emptyset}}
\{\setunsat(\calC,S)\}.
\end{align}

We say that $\calC$ is satisfiable if  $\setunsat(\calC)=0$ and for $\eps > 0$, we
say that 
$\calC$ is $\eps$-frustrated if 
$\setunsat(\calC)\geq \eps$.
\end{definition}

Notice that the normalization factors in \Cref{unsat} guarantees that the $\setunsat$ value lies between  $0$ and $1$.\footnote{The lower bound is trivial since the value cannot be negative. For the upper-bound, notice that $B_C,L_C \subseteq S$ and $B_C \cap L_C = \emptyset$. This is because bad strings have no neighbors whereas longing strings do. Hence  
$ \frac{|B_C|}{|S|} + \frac{|L_C|}{|S|} \leq  \frac{|S|}{|S|} = 1$.} 
{~}

\subsection{Intuition and standard CSP as special case}
 We can present a standard \kCSP{} instance consisting of constraints 
 $C_1,C_2,...,C_m$
 as an instance of $\kSetCSP$ in the following way: For each
 constraint $C_\ell, \ell\in\{1,...,m\}$ out of the $m$ constraints in the $\kCSP$ instance, 
 we consider the following set-constraint $C_\ell'$: for 
 every $k$-bit string $s$ that {\em satisfies} $C_\ell$, add the subset $Y_s = \{s\}$ to $C_\ell'$.  
 We arrive at a collection $\calC$ of $m$ set-constraints, where each set-constraint $C_\ell'$ in $\calC$ consists of single-string sets 
 $Y_i$ 
 corresponding to all strings which satisfy $C_\ell$.

We claim that the resulting \SetCSP{} instance has $\setunsat(\calC)=0$
if and only if the original \CSP{} instance was satisfiable. 
 First, if the original \CSP{} instance is satisfiable, we claim that for any 
 satisfying string $s$ we can define the set $S=\{s\}$ consisting of that single string, and $S$ indeed satisfies the collection of set-constraints defined above. 
To see this, note that in our case, there is no notion of $C$-neighbors (See Definition \ref{def:neighbors}), since all $Y_i$'s contain only a single string. Hence, there are no longing strings; 
 By the definition of $\setunsat$ (\Cref{unsat}) in this case 
  $\setunsat(\calC,S)=0$ if all strings in $S$ are good for all
 set-constraints $C'$, namely each of them satisfies all 
 constraints $C$ in the original $\kCSP$ instance, which is indeed the case if we pick a satisfying assignment. For the other direction, assume $\setunsat(\calC,S)=0$ for some set $S$. This in particular means that all strings in $S$ are $C$-good for all set-constraints in $\calC$. By definition of our $\calC$, this means any string $s\in S$ is a satisfying assignment.

We give now some intuition about the $\setunsat$ quantity (\Cref{unsat}). 
We first note that it generalizes the by-now-standard notion of (un)satisfiability in \CSP{}, 
which for a given string, counts the number of unsatisfied  constraints, divided
by $m$. We note that in \Cref{unsat}, 
if $S = \{s\}$, then $|B_C|$ is either 
$0$ or $1$, depending on whether $s$ satisfies the constraint or not. As previously remarked, in \CSP{} the notion of neighboring and longing strings does not exist, and so $L_C$ will always be empty.
Thus, in the case where $S = \{s\}$ and all set-constraints containing single strings, \Cref{unsatave}  is indeed the number of violated constraints by the string, divided by $m$ -- which is exactly the standard CSP \textsf{unsat} used in PCP contexts~\cite{Dinur07}. (In the case of $S$ containing more than one string, 
but all $Y_i$s are still single-strings,  \Cref{unsatave} will just be  
the (non-interesting) average of the \textsf{unsat} of all strings in $S$).

The interesting case is the general \SetCSP{} case, when the $Y_i$'s contain more than a single string, i.e., when the notions of neighbors and longing strings become meaningful. 
In this case,
the left term in the RHS of \Cref{unsat} quantifies how far the set $S$ is from the situation in which it contains only ``good" (not ``bad") strings (this can be viewed as the standard  requirement) but it adds to it the right term, 
which quantifies how far $S$ is from being closed to 
the action of adding neighbors 
with respect to the set-constraints\footnote{Notice that we could have chosen to define "far" here, by counting  
the number of strings {\it outside} of $S$ 
that are neighbors of $S$. But since the degree of each 
string in the graph is bounded, the exact choice of the definition does not really matter; we chose the one presented here since it seems 
most natural.}.
Loosely put, the generalization from constraints to set-constraints 
imposes strong ``dependencies" between 
different strings, 
and the number of longing strings, $L_C$, counts to what 
extent these dependencies are violated. 

\subsection{\MA-hardness of $\SetCSP{}_{1/\poly{n}}$}
In this subsection, we show the \MA{}-hardness of $\SetCSP_{1/poly}$.  

\begin{lemma}\label{cor:ma-hardness-setcsp}
There exists some inverse polynomial $p(x) = \Theta(1/x^3)$ such that for every inverse polynomial $p' < p$, the problem $\SetCSP_{p'(m)}$ is $\MA$-hard.
\end{lemma}

To prove this lemma, we show how to reduce any language $\textsf{L} \in \MA$ to a $6$-\SetCSP{} instance. Our approach here is to ``mimic'' the Quantum Cook-Levin theorem due to Kitaev~\cite{KitaevSV02}, but given that we only need to
deal with set of strings and not arbitrary quantum states, our proof can in fact be stated in set-constraints language.

\subsubsection{Intuition for how set-constraints can check histories}  \label{sec:intuition}
In the celebrated proof of the (classical) Cook-Levin theorem, an instance of an \NP-language is mapped  to an instance for $3$-$\problem{SAT}$ problem. More precisely, the verifier 
$V$ of the \NP-problem, which runs on an $n$-bit input string $x$ and a $poly(n)$ bit witness $y$, is mapped to a $3$-$\problem{SAT}$ formula. 
To do this, a different variable is assigned to the value of each location of the tape of the Turing machine of $V$ at any time step; these variables are used to keep track of the
state of the computation of $V$ (namely what is written on the tape) at the different time steps (see \Cref{fig:comparison} for an example). The formula acts on strings of bits, which can be viewed as assignments to {\it all} these variables; such an assignment encodes the {\it entire history} of a 
single possible computation. The clauses of the
Boolean formula check if the history given by the assignment, indeed corresponds to a correct propagation of an {\it accepted} computation. More precisely, the constraints check that $a)$ the assignment to the 
$n$ Boolean variables at the first time step,  
associated with the input to the \NP-problem, is indeed correct (namely equal to the true input $x$); $b)$ the assignment of the variables corresponding to any two subsequent time steps is consistent with
a correct evolution of the appropriate gate in the computation; and $c)$ the output bit, namely the value of the output variable in the {\it last} time step, is indeed {\it accept}. There exists a
valid history which ends with accept (i.e., $x$ is in 
the language accepted by the verifier), iff the resulting \kCSP{} is satisfiable; in which case a satisfying assignment encodes a  
{\it history} of a correct computation of the verifier.

Now, suppose we want to apply a similar construction for some
\MA{} verification. The random bits are also given to the verification circuit as an input, and one could hope that the reduction of the Cook-Levin theorem would still work. However, the problem is that there could be {\it some} choice of random coins that makes the verification algorithm accept even for no-instances, because the soundness parameter is not $0$. Hence, the \CSP{} would be satisfiable in this case, even though the input is a NO instance. It is thus not sufficient to verify the existence of a {\it single} valid history. In order to distinguish between YES and NO instances of the problem, we have to check in the YES case that {\it all} (or at least many of the) initializations of the random bits lead to accept. The key difficulty is how to check that not only a single string
satisfies the constraints, but many.  

This is exactly the reason for introducing the {\em set-constraints}. These set-constraints are able to verify that the 
random bits are indeed {\it uniformly random}; then the standard Cook-Levin approach described above can verify that (given the right witness)
most of them leads to acceptance. 

In order to implement such an approach, 
we first explain how to modify the original Cook-Levin 
proof, of the \NP{} completeness of \textsf{satisfiability}, 
so that the strings that we check are not entire 
histories of a verification process of some \NP{} problem; 
rather, the strings represent 
{\it snapshots} of 
the tape (or evaluations of all bits involved in the circuit at a certain time) for {\it different time-steps} of the computation. A satisfying assignment is no longer
a single string but a whole collection of strings, denoted $S$,
which would be the {\it collection of all snapshots 
of a single valid computation, at different times}.\footnote{ 
For quantum readers, we note that this reflects the main step in Kitaev's modification of the Cook-Levin theorem, which enabled him to test that entangled quantum states 
evolved correctly.} 

In this case (note that we have still not included random bits in the 
discussion) we need to show how to create {\em set-constraints} that verify that the set $S$ really does contain {\it all} snapshots, and it also needs to verify that $S$ contains nothing else. 
More precisely, we need to verify that 
$a)$ the strings in $S$ are consistent with being snapshots of a valid evolution of the computation 
$b)$ the input is correct in the string corresponding to the first snapshot, and the output is accept in the string corresponding to the last snapshot, and $c)$ the set $S$ indeed contains the whole {\em
history} of the computation, i.e. all the snapshots in a correct evolution, and without any
missing step. 

It turns out that this can be done using set-constraints;  
We depict the differences between the ``original'' Cook-Levin proof and the ``set-constraints'' one in \Cref{fig:comparison}.

In reality, we need to do this not just for a single evolution but 
for the evolution over all possible assignments to the random strings. Moreover, we need to enforce that the random bits are indeed random; this requires further set-constraints (technically, this is done below in \Cref{eq:constraint-random}).  

We note that the essence of the translation idea mentioned in the introduction appears already when there are no random bits involved at all; the reader is recommended to pretend that no random bits are used, at her first reading. 

We next explain the reduction from \MA-verification algorithms to $6$-\SetCSP{} in
\Cref{sec:reduction-proof}, and then prove its correctness in
\Cref{sec:correctness}.

\subsubsection{The reduction}\label{sec:reduction-proof}

We assume, without loss of generality, that the \MA{}
verification algorithm is reversible, as described in \Cref{sec:reversible}:
Given an input $x \in \{0,1\}^n$, we assume a verification circuit $C_x$ whose input is
$y\concat 0^{a(n)} \concat r$, where $y  \in \{0,1\}^{p(n)}$ is the polynomial-size \MA{}
witness, the middle register consists of the circuit auxiliary bits (needed for reversibility)
and $r \in \{0,1\}^{q(n)}$ are the polynomially many random bits used during the
\MA{} verification.  The circuit $C_x$ consists of $T$ gates  $G_1,...,G_{T}$,  where
$G_i \in \{\NOT, \CNOT, \CCNOT\}$.\footnote{Recall that the gate $\CNOT$ on input $a,b$ outputs $a,b \oplus a$, and  the gate $\CCNOT$ on input bits $a,b,c$, outputs $a,b,c\oplus a b$.}   At the end of the circuit, we assume WLOG that the first bit
as the output bit.
\begin{figure}[t]
  \begin{subfigure}[b]{0.42\textwidth}
\begin{adjustbox}{max totalsize={0.9\textwidth}{\textheight},center}
\begin{circuitikz} \draw
(-2,0) node {}
(0,2) node[and port] (myand1) {\hspace{-0.6em}\Large$\wedge$}
(0,0.7) node[and port] (myand2) {\hspace{-0.6em}\Large$\wedge$}
(0,-0.6) node[and port] (myand3) {\hspace{-0.6em}\Large$\wedge$}
(0,-1.9) node[and port] (myand4) {\hspace{-0.6em}\Large$\wedge$}
(2.5,1.2) node[or port] (myor1) {}
(2.5,-1.3) node[or port] (myor2) {}
(5,0.1) node[and port] (myand5) {\hspace{-0.6em}\Large$\wedge$}
(myand1.in 1) node[left=.5cm](x1) {$y_1$}
(myand1.in 2) node[left = .5cm](x2) {$y_2$}
(myand2.in 1) node[left=.5cm](x3) {$y_3$}
(myand2.in 2) node[left = .5cm](x4) {$y_4$}
  (myand3.in 1) node[left=.5cm](x5) {$y_{5}$}
  (myand3.in 2) node[left = .5cm](x6) {$y_{6}$}
  (myand4.in 1) node[left=.5cm](x7) {$y_{7}$}
  (myand4.in 2) node[left = .5cm](x8) {$y_{8}$}
  (-2.25,2.6) node  (dummy0) {}
  (myand1.out)  node[above right](x9) {$y_{9}$}
  (myand2.out)  node[above right](x10) {$y_{10}$}
  (myand3.out)  node[above right](x11) {$y_{11}$}
  (myand4.out)  node[above right](x12) {$y_{12}$}
  (myor1.out)  node[above right](x13) {$y_{13}$}
  (myor2.out)  node[above right](x14) {$y_{14}$}
(myand5.out)  node[right](x15) {$y_{15}$}
(myand1.out) -| (myor1.in 1)
(myand2.out) -| (myor1.in 2)
(myand3.out) -| (myor2.in 1)
(myand4.out) -| (myor2.in 2)
(myor1.out) -| (myand5.in 1)
(myor2.out) -| (myand5.in 2)
(x1) -| (myand1.in 1)
(x2) -| (myand1.in 2)
(x3) -| (myand2.in 1)
(x4) -| (myand2.in 2)
(x5) -| (myand3.in 1)
(x6) -| (myand3.in 2)
(x7) -| (myand4.in 1)
(x8) -| (myand4.in 2);
  \draw[color=red,ultra thick,opacity=0.5]  (dummy0) --(x8.-90);
\draw[dashed,color=red,ultra thick,opacity=0.5]     (x8) to[out=-90,in=90] (x9);
  \draw[color=red,ultra thick,opacity=0.5]  (x9.90) -- (x12.-90);
  \draw[dashed,color=red,ultra thick,opacity=0.5]    (x12) to[out=-90,in=90] (x13);
  \draw[color=red,ultra thick,opacity=0.5]  (x13.90) --(x14.-90);
  \draw[dashed,color=red,ultra thick,opacity=0.5]    (x14) to[out=-90,in=90] (x15);
  \draw[color=red,ultra thick,opacity=0.5]  (x15.90) --(x15.-90);
  \draw node at (1.7,-3){$h = 011111010110111$};
\end{circuitikz}
    \end{adjustbox}
    \caption{In the Cook-Levin theorem, a fresh new variable is assigned to
    every wire of the circuit, and the evolution of the computation can be
    described as an assignment to the variable such that their values are
    consistent according to the circuit. In this example, we see a (very simple)
    circuit, and then the history of the computation on input $01111101$. The
    \CSP{} instance derived from the Cook-Levin theorem ensures that the
    assignment is indeed the evolution of the circuit for some input (and of
    course, that the circuit accepts at the end).
    \vspace{3.8em}
    }
    \label{fig:cook-levin}
  \end{subfigure} \quad\quad
  \begin{subfigure}[b]{0.50\textwidth}
        \begin{equation*}
        \Qcircuit @C=0.8em @R=.8em {
                       & s_0  &           & s_1 &                    & s_2  & & s_3  &  \\
          \lstick{y_1} & \qw  & \qw       & \qw & \multigate{2}{\text{\footnotesize CCNOT}}& \qw  & \qw                  & \qw  & \qw & \rstick{\text{output}} \\
          \lstick{y_2} & \qw  & \qw       & \qw & \ghost{CCNOT}       & \qw  & \multigate{2}{\text{\footnotesize CCNOT}} & \qw  & \qw\\
          \lstick{y_3} & \qw  & \gate{NOT}& \qw & \ghost{CCN}                & \qw  & \ghost{CCN}        & \qw  & \qw\\
          \lstick{y_4} & \qw  & \qw       & \qw & \qw                & \qw  & \ghost{CCN}        & \qw  & \qw
          \gategroup{2}{2}{5}{2}{.7em}{--}
          \gategroup{2}{4}{5}{4}{.7em}{--}
          \gategroup{2}{6}{5}{6}{.7em}{--}
          \gategroup{2}{8}{5}{8}{.7em}{--}
        }
        \end{equation*}
        {\footnotesize
\[H = \{\textbf{000}1110,\textbf{100}1100,\textbf{110}1110,\textbf{111}1111\}\]
    }
    \\~\\
    \caption{In this work, we consider {\em reversible} circuits and the history
    of the computation is described by a {\em set of strings}.
    We define then the {\em history-set} of the computation that contains the {\em
    snapshot} of every stage of the computation. Each such a string is augmented with a prefix (marked in bold) identifying the timestep, in unary, of the snapshot. 
    These bits, called the ``clock", are necessary in order to check the evolution.    
    In this example, we see a very simple
    circuit that has no auxiliary nor random bits.
We show the history-set of the computation on input $1110$ and the
    prefix indicating the number of the timestep (here, $0$ to $3$), counted in unary.
    The
    \SetCSP{} instance derived here ensures that a satisfying set of strings 
    contains the snapshots for {\em all} timesteps of the computation.}
    \label{fig:reversible-cook-levin}
  \end{subfigure}
\caption{Comparison between the evolution of a circuit in the Cook-Levin
proof and in our work.}
\label{fig:comparison}
\end{figure}
We describe now the reduction from the \MA{} problem  into a \SetCSP{} instance $\calC$. We will show in the
next section that if there is an \MA{} witness that makes the verification
algorithm accept with probability $1$, then $\calC$ is satisfiable, whereas if
every witness makes the verifier reject with probability at least $\frac{1}{2}$,
then $\calC$ is at least inverse polynomially frustrated.

We start with an MA verification circuit $C_x$ (assumed to be reversible, as described in \Cref{sec:reversible})  
acting on an 
input consisting of $a(n)$ $0$-bits, witness $y$ of $p(n)$ bits and $q(n)$ random bits. Thus, the number of bits which the reversible verification circuit acts on is $w(n)=a(n)+p(n)+q(n)$. 
The number of gates is $T(n)$; denote these gates by 
$G_1,...,G_T.$
Our set-constraint instance $\calC$ will act on strings of   $s(n) = T(n) + w(n)$ many bits. 
We omit $n$ from such functions from now on, since it will be clear from the context.

We call the $T$ first bits of such strings the {\it clock} register and the
last $w$ bits the {\it work} register. The work register comprises of three sub-registers: the witness register (first $p(n)$ bits), the auxiliary register (middle $a(n)$ bits) and the randomness register (last $q(n)$ bits). We want to create
set-constraints which force all the strings to be of the form  $z \in \01^{s}$ such that  $z|_{[T]} = \unary(t)$ for some $t
\in [T+1]$ (where $\unary(t)$ denotes the integer $t$ written in unary representation), 
and $z|_{[s] \setminus [T]}$ represents the {\em
snapshot} of the computation at time $t$ (as explained in \Cref{sec:intuition} above) for the initial string 
$y\concat 0^{a(n)} \concat r$, 
for some witness $y$ and
some choice of random bits $r$.

We will construct $\calC$ in such a way 
that $S \subseteq \{0,1\}^s$ satisfies $\calC$ if and only if $i)$ it contains {\it all} snapshots of 
the computation for some witness $y$ and for {\it all} bit strings $r$ input to the randomness register; and $ii)$ the
computations whose snapshots are contained  in $S$ are not only correct (meaning also that the auxiliary bits are all initialized to $0$) but that they are {\it accepted} computations (meaning that 
the output is $1$). 

We do this by providing set-constraints of four types, as follows. 

\paragraph{Clock consistency.} We first impose that if $z \in S$, then $z|_{[T]} =
    \unary(t)$, for some $t \in [T+1]$.

Notice that a string is a valid unary encoding
iff it does not contain $01$ as a substring.
To guarantee that the clock bits are consistent with some 
unary representation of an allowed $t$, we add 
 for every $t\in [T]$ the
set-constraint
 $C^{clock}_{t}$ defined by: 
 \[Y(C^{clock}_{t}) = (\{\{00\},\{10\},\{11\}\}) \text{ and } J(C^{clock}_{t}) = (t, t+1).\]

\noindent{\underline{Example:}} The string $010^{T-2} \concat z$ is a bad
 string for $C^{clock}_1$ since $w|_{J(C^{clock}_1)} = 01 \not\in
 Y(C^{clock}_1)$.

\paragraph{Initialization of Input bits and Random bits.}  Here, we want to check that $0^T\concat y \concat z \concat r$ is not in $S$ whenever $z \ne 0^a$, which enforces the ancillary bits to be initialized to $0$. In addition, we want to check that for any witness $y$,   if one string of the form $0^T\concat y \concat 0^a \concat r$ for some $r$ is in $S$, then for all $r' \in \01^q$, we have $0^T\concat y \concat 0^a \concat r'$ in $S$. 
In conclusion, we need to check two things: that all auxiliary bits are
initialized to $0$ and that all possible initializations of the random bits are present.

For each {\it auxiliary bit} $j \in [a]$ we add a set-constraint $C^{aux}_j$ and for
every random bit $j \in [q]$ we add the set constraint $C^{rand}_j$ as follows.

We define 
\[Y(C^{aux}_j) = \{\{00\},\{10\},\{11\}\} \text{ and } J(C^{aux}_j) =
(0,T+p+j),\]
which forces that for $t = 0$ (notice that the unique value of $t$ for which 
$\unary(t)$ has the first bit $0$ is $t = 0$), the $j$-th auxiliary bit must be
$0$, because the string $(01)$ is forbidden. For $t \ne 0$ this is not enforced by allowing any value of the $j$-th
auxiliary bit when the first clock bit is $1$ (and therefore $t \ne 0$).

\smallskip
\noindent{\underline{Example:}} The string $0^T \concat y \concat 10^{a-1} \concat r \in S$
is bad for $C^{aux}_0$.
\medskip

For the {\it random bits}, we want to make sure that the $j$-th
random bit has both values $0$ and $1$, over all the random bits.  Therefore 
we define the constraints $C^{rand}_j$ by 
\begin{align}\label{eq:constraint-random}
    Y(C^{rand}_j) = \{\{00,01\},\{10\},\{11\}\} \text{ and } J(C^{rand}_j) =
(0,T+p+a+j).
\end{align}  

These constraints will be useful in the following way. 
First, the propagation set-constraints that we will define soon, constrain $S$ to make sure that there {\it exists} 
a string $0^T \concat y \concat 0^a \concat r$ representing a snapshot at time $0$
with some value of the random bits, $r$, which is indeed in $S$. The $C^{rand}$ constraints enforce that given the existence of such a string in $S$, then for any other assignment to the random bits, $r'$, the string $0^T \concat y \concat 0^a \concat r' \in S$. Then, if many of these other strings are not in $S$, the frustration 
will be high.

\noindent{\underline{Example:}} 
If 
$s_1 = 0^T \concat y \concat 0^a \concat 0^r \in S$ 
but 
$s_2 = 0^T \concat y \concat 0^a \concat 10^{r-1} \not\in S$, then $s_1$ is a
$C^{rand}_1$-longing string in $S$ since
$s_1$ and $s_2$ are $C^{rand}_1$-neighbors. 

\paragraph{Propagation.} Here we want to check that 
if $\unary(t-1) \concat z \in S$ for some $0< t < T$,
    then $\unary(t) \concat G_{t}\left(z\right) \in S$.

Let us consider the propagation constraint associated with the $t$-th timestamp (the one corresponding to the application of the $t$th gate, $G_t$), for $1<t<T$. Let
us assume that the $t$-th gate acts on bits $b_{t,1},...,b_{t,k}$. For
this we add the set-constraint $C^{prop}_t$ defined as follows:
\[Y(C^{prop}_t) = \bigcup_{z\in \01^{k}} \{\{100\concat z,110\concat G_t(z)\}\}
\text{ and  } J(C^{prop}_t) = (t-1,t,t+1,b_{t,1},b_{t,2},...,b_{t,k}).\]
For $t=1$ we simply erase the left clock bit from the
above specification: 
\[Y(C^{prop}_1) = \bigcup_{z\in \01^{k}} \{\{00\concat z,10\concat G_1(z)\}\}
\text{ and  } J(C^{prop}_1) = (1,2,b_{1,1},b_{1,2},...,b_{1,k}).\]
Likewise if $t=T$ erase the right most clock bit from the above. 

\noindent{\underline{Example:}} 
If 
$s_1 = \unary(t) \concat z \in S$ 
but 
$s_2 = \unary(t+1) \concat G_{t+1}(z) \not\in S$, $s_1$ is a $C^{prop}_{t+1}$-longing string in $S$,
 since the two strings are $C^{prop}_{t+1}$- neighbors.

\paragraph{Output.} Finally, we need to check that for all strings of the form $1^T\concat z$, 
the first bit of $z$ is $1$, namely, the last snapshot 
corresponds to accept. 
We define $C^{out}$ such that 
\[ Y(C^{out}) = \{\{00\},\{01\},\{11\}\} \text{ and } J(C^{out}) = (T,T+1).\] 
Here we use the fact
that the $T$-th bit of $\unary(t)$ is $1$ iff $t = T$. In this case, if this
value is $1$, we require that the output bit is $1$, otherwise it could have
any value.

\noindent{\underline{Example:}} 
The string $1^T \concat 0 \concat z$ is bad for $C^{out}$.

\begin{remark}
We notice that for every set-contraint $C$ that we constructed above, we have that $Y(C)$ only contains sets of size $1$ or $2$. This property adds a bit more structure to $\SetCSP$ instances that are \MA{}-hard, which could be useful in future work.
\end{remark}

\newcommand{\soundness}{\ensuremath{\frac{1}{10(T+1)q m}}}

\subsubsection{Correctness}
\label{sec:correctness}
Given some \MA-verification circuit $C_x$, we consider the following $6$-\SetCSP{} instance 
\[\calC_x =
(C^{clock}_1,....,C^{clock}_T,C^{aux}_1,...,C^{aux}_a,
C^{rand}_1,...C^{rand}_q,C^{prop}_1,...,C^{prop}_T,C^{out}\}.\] 

Let $m$ be the number of terms in $\calC_x$ and when it is more convenient to us, we will refer to the set-constraints in $\calC_x$ as $C_i$ for $i \in [m]$, where the terms have an arbitrary order. We show now that $\calC_x$
is satisfiable if
$x$ is a positive instance, 
and if $x$ is a negative instance, then $\calC$ is at least \soundness-frustrated (where we remember that $q$ is the number of random coins used by $C_x$).
Notice that $m$, the number of constraints in $\calC_x$, is polynomial in $T$, which is
also polynomial in $|x|$.

\medskip

We start by proving completeness.
\begin{lemma}[Yes-instances lead to satisfiable $\SetCSP$ instances]\label{lem:hardness-completeness}
  If $x \in \mathsf{L}$, then $\mathcal{C}_{x}$ is satisfiable.
\end{lemma}
\begin{proof}
Let $y$ be the witness that makes $C_x$ accept with probability $1$, and let 
  \[S = \left\{\unary(t) \concat G_t...G_1(y,0^{a},r) : r \in \01^q, t \in
  [T+1]\right\}.\]

By construction, the initialization, clock and propagation constraints are
  satisfied by $S$. By the assumption that the MA verification circuit accepts
  with probability $1$, the output constraints are also satisfied by $S$.
\end{proof}

Next we prove soundness.
\begin{lemma}[NO-instances lead to frustrated $\SetCSP$ instances] \label{lem:hardness-soundness}
  If $x \not\in \mathsf{L}$, then $\setunsat(\mathcal{C}_{x},S) \geq \soundness $
  for every non-empty $S \subseteq \01^n$.
\end{lemma}
\begin{proof}
  Let $S \subseteq \{0,1\}^{n}$ be a non-empty set, $B$ the set of bad strings in $S$ (namely a string in $S$ which is $C$-bad for at least one set-constraint $C$) and $L$ the set of longing strings in $S$ (namely the strings in $S$ which are $C$-longing for at least one set-constraint $C$).
Our goal here is to consider a partition $\{K_i\}$ of $S$, such that 
for every $K_i$, $|K_i\cap (B\cup L)| \geq \frac{|K_i|}{10(T+1)q}$,
and from this we will show that $\setunsat(\calC_x,S) \geq
  \soundness$.

Let us start by defining $S_{ic} \subseteq S$ to be the subset of $S$ with \textbf{i}nvalid
  \textbf{c}lock register. Notice that every $x \in S_{ic}$ is bad for at least one
  clock constraint and therefore $|S_{ic} \cap B| = |S_{ic}|$. 

Now we notice that all other strings correspond to some valid clock register whose value (when read as a unary representation of some integer) is in $[T+1]$. 
Let us partition the strings in $S \setminus S_{ic}$ into disjoint sets $H_1$,...,$H_\ell$ (which indicate different history-sets to which the strings belong) as follows. 
We define the {\it initial configuration} of some string
in  $S \setminus S_{ic}$ like this. The string must be of the form $\unary(t) \concat z$ for some $t$ and $z$. Then $\initial(\unary(t) \concat z) = \unary(0) \concat G^{-1}_1...G^{-1}_t(z)$, which is the assignment of the initial bits  that leads to the configuration $z$ at the $t$'th step. We say then that two strings $s_1, s_2 \in H_i$ iff $\initial(s_1) = \initial(s_2)$ and we abuse the notation and call $\initial(H_i) = \initial(s_1)$. Notice that for $i \ne j$, $H_i$ and $H_j$ are disjoint because the computation is reversible; thus the $H_i$'s constitute a partition of $S\setminus S_{ic}$. 
  Notice also that each $H_i$ contains at most
  $T+1$ strings, and the different strings in each $H_i$ have  different values of the clock register. 
We call these $H_i$ {\it history-sets} for the reason that they correspond to a correct propagation of the computation of the circuit $C_x$ for some initialization of all its bits.\footnote{
More precisely, $H_i$ is a subset of the correct propagation because it involves only the snapshots in $S$. 
}   

Let us first consider the history-sets whose  initial configuration is not valid, i.e., it contains \textbf{i}nvalid (or non-zero) \textbf{a}uxiliary bits: 
$\textbf{H}^{ia} = \{H _i :
\initial(H_i) = 0^T \concat y \concat z \concat r\text{ for some } z \ne 0^a\}$. 
We note that for any $H_i \in \textbf{H}^{ia}$, $\initial(H_i)$ is a $C^{aux}_j$-bad string in $H_i$ for some $j$.
If this string is in 
$H_i$, then we indeed have 
$|H_i \cap (B\cup L)|\geq |H_i \cap B| \geq 1
  \geq \frac{|H_i|}{T+1}$.
  However, if $H_i$ does not 
  contain its initial string, 
  then consider the minimal 
  $t$ such that $\unary(t)\concat z\in H_i$ for some $z$, and by assumption we have $t>0$.  This means that  the string 
  $\unary(t) \concat z$
 is a $C^{prop}_t$-longing string, because it is a neighbor of 
    $\unary(t-1) \concat G_t^{-1}(z) \not\in S$.
Hence, for such $H_i$ we have  
$|H_i \cap (B\cup L)|\geq |H_i \cap L| \geq 1
  \geq \frac{|H_i|}{T+1}$. This completes handling all the strings
in $S$ within a history set $H_i$ with invalid 
auxiliary bits in $\initial(H_i)$.  

We now need to consider 
history sets in $\{H_i\} \setminus \textbf{H}^{ia}$, namely, the history sets whose initial string is of the form 
$0^T \concat y \concat 0^a \concat r $ for some value of $y$ and $r$. 
Let us group these history sets 
according to $y$, the value of the witness register.
In other words, let us consider the sets of history sets: 
  $\textbf{H}_y = \{H_i : H_i\text{'s initial string is of the form } 0^T \concat y \concat 0^a \concat r \}$. 
We fix now some $y$ and the following arguments hold for each $y$ separately. Notice that each $H_i \in \textbf{H}_y$ corresponds to a computation
  corresponding to a different initial random string for the witness $y$. 
Let us denote the union of strings in all sets in $\textbf{H}_y$ by $\textbf{S}_y = \bigcup_{H_i \in \textbf{H}_y} \{s|s\in H_i\}$. 
To finish handling all strings, we need to provide a bound 
 on $|\textbf{S}_y \cap (L\cup B)|$ for all witnesses~$y$.

To proceed, we need some additional notation. We note that $\textbf{H}_y$ can be written as the union of the history sets which contain their initial string, denoted 
$\textbf{H}_y^{start}$, and the rest, denoted 
 $\textbf{H}_y^{nostart}$. 
 We proceed by considering two cases separately: 
 
\begin{enumerate} 
\item 
Let us first consider the simpler case in which  $|\textbf{H}_y^{nostart}|>
\frac{|\textbf{H}_y|}{10}$. 
As above, we have that each $H_i \in
\textbf{H}_y^{nostart}$ has a longing string, and therefore we have that 
  $|\textbf{S}_y \cap (L\cup B)| \geq |\textbf{S}_y \cap L| \geq |\textbf{H}_y^{nostart}| > \frac{|\textbf{H}_y|}{10}\ge \frac{|\textbf{S}_y|}{10(T+1)}$. The last inequality is due to the fact that 
  each $H_i\in\textbf{H}_y$ contains at most $T+1$
  strings. This finishes the treatment of 
  all strings in $S$, in the case 
  $|\textbf{H}_y^{nostart}| >
\frac{|\textbf{H}_y|}{10}$.
  
\item  
In the second case 
$|\textbf{H}_y^{start}|\ge
\frac{9|\textbf{H}_y|}{10}$. 
 We further denote  
$\textbf{S}_y^{init} = \{s | s = initial(H_i), H_i \in \textbf{H}_y^{start} \}$ as the set of the initial strings of each $H_i \in \textbf{H}_y^{start}$. 
Again (and for the last time) there are two cases. 
\begin{enumerate} 
\item First, let us consider the case when  $|\textbf{S}^{init}_y| \geq 2^{q-1}$. This means
that for this fixed $y$, for most values $r$ of the random bits, the initial string $0^T\concat y \concat 0^a \concat r$ of the history set corresponding to this $y$ and $r$
is present in $S_y$.  
We use the facts that $x\notin \mathsf{L}$, and that at least $2/3$ of the history sets must lead to rejection by \Cref{def:MA}. From these observations, we will conclude that there will be either  many bad strings due to the final accept 
constraint $C^{out}$, or many longing strings.

 Let  $\textbf{Acc}_y
  = \{H_i \in \textbf{H}_y : 1^T \concat z \in H_i \text{ and } z = 1\concat z' \}$ be the set of
  history sets in $\textbf{H}_y$ that accept in the last step. We have that $|\textbf{Acc}_y|$ is at most the
  number of $r\in \{0,1\}^q$ which leads the circuit $C_x$ to accept the witness $y$; since $x \not\in \mathsf{L}$, we have that the 
  probability to accept for any $y$ is most $1/3$. Hence 
  $|\textbf{Acc}_y|\le \frac{2^q}{3} \leq 2\frac{|\textbf{H}_y|}{3}\leq \frac{20|\textbf{H}^{start}_y|}{27}$, where we used the fact that
  $2^{q-1} \leq  |\textbf{S}_y^{init}| = |\textbf{H}^{start}_y| \leq |\textbf{H}_y|$, and in the last inequality, the fact that we are in the case $|\textbf{H}_y^{start}|\geq 9|\textbf{H}_y|/10$. 
  Hence, there are at least 
  $\frac{2^q}{10}$ history sets 
  in $\textbf{H}^{start}_y$ which do not end in accept. Such $H_i$
   either contains the string $1^T \concat 0 \concat z$ (which is bad for $C^{out}$), or does not contain a final state at all, namely does not contain a state of the form $1^T\concat z$ for some $z$, resulting in a longing string.
   Thus, there are at least 
   $\frac{2^q}{10}$ strings in 
   $|\textbf{S}_y \cap (B\cup L)|$; 
   and since $|\textbf{S}_y| \leq (T+1)|\textbf{H}_y|$ and $|\textbf{H}_y| \leq 2^q$, resulting in $|\textbf{S}_y| \leq 2^q(T+1)$, 
   we have that
   $|\textbf{S}_y \cap (B\cup L)|
   \ge \frac{2^q}{10}\ge \frac{|\textbf{S}_y|}{10(T+1)}.$

\item Finally, let us consider the case where  $|\textbf{S}_y^{init}| \leq 2^{q-1}$.
This is where we will need to apply conductance 
arguments. 
Let $G^{0}_{y}$  be the subgraph of $G_{\mathcal{C}}$\footnote{Here, we use the same notation as \Cref{S:ma}.} induced by the vertices $R_y = \{0^T\concat y \concat 0^a \concat r : r \in \{0,1\}^q\}$. Notice that $G^{0}_{y}$ is isomorphic to the $q$-dimensional hypercube. This is true because the only remaining edges on $G^{0}_y$ come from the set-constraints $C_j^{rand}$.
Notice also that $\textbf{S}_y^{init}$ is a subset of the vertices of $G^{0}_{y}$. We now apply \Cref{lem:expansion-hypercube} which states 
that the conductance of the $q$-dimensional hypercube is $\frac{1}{q}$. Applying this lemma to the graph 
$G^{0}_{y}$ and the subset of its vertices, $\textbf{S}_y^{init}$ we conclude that  
$\frac{|\partial_{G^0_y}(\textbf{S}_y^{init})|}{q|\textbf{S}_y^{init}|}\ge \frac{1}{q}$, where we have used the fact that 
all vertices in $G^{0}_{y}$
have the same degree, $q$, and the fact that 
we are now considering the case $|\textbf{S}_y^{init}| \leq 2^{q-1}$.
We can conclude then that there exists at least $|\textbf{S}_y^{init}|$ edges in the cut $(\textbf{S}_y^{init}, R_y \setminus \textbf{S}_y^{init})$. Since each vertex in $G^{0}_{y}$ (in particular, each vertex in $\textbf{S}_y^{init}$) has $q$ neighbors, it means that there exists at least $\frac{|\textbf{S}_y^{init}|}{q}$ longing strings in $\textbf{S}_y^{init}$. We conclude this case by noticing that 
\[|\textbf{S}_y\cap(L\cup B)| \geq  |\textbf{S}_y\cap L|\geq \frac{|\textbf{S}_y^{init}|}{q} = \frac{|\textbf{H}^{start}_y|}{q} \geq \frac{9|\textbf{H}_y|}{10q} \geq  \frac{9|\textbf{S}_y|}{10q(T+1)}, \] 
where in the second inequality we use the fact that $\textbf{S}_y^{init} \subseteq \textbf{S}_y$ and there are at least $\frac{|\textbf{S}_y^{init}|}{q}$ longing strings in $\textbf{S}_y^{init}$, the equality follows since  $|\textbf{H}_y^{start}| = |\textbf{S}_y^{init}|$, the third inequality follows from our assumption that  $|\textbf{H}_y^{start}|\ge
\frac{9|\textbf{H}_y|}{10}$  and finally we have that  $|H_i| \leq T+1$  (and therefore $|\textbf{H}_y| \geq \frac{|\textbf{S}_y|}{T+1}$).
\end{enumerate}
\end{enumerate} 
  
  To finish the proof, notice that $S = S_{ic} \cup \bigcup_{H \in
  \textbf{H}^{ia}} H \cup \bigcup_{y} \mathbf{S}_y$. Since each of these subsets has at least a $\frac{1}{10(T+1)q}$-fraction of bad strings or longing strings, we have that $|B \cup L| \geq \frac{|S|}{10(T+1)q}$. It follows that  
  \[
   \setunsat(\calC_x,S) = \frac{1}{m} \sum_{i} \left( \frac{|B_{C_i}|}{|S|} + \frac{|L_{C_i}|}{|S|} \right) \geq  \frac{|B|}{m|S|} + \frac{|L|}{m|S|}
   \geq \frac{|B\cup L|}{m|S|} \geq   \frac{1}{10(T+1)q m},
  \]
  finishing the proof.
\end{proof}

From the two previous lemmas, we have the following.

\newtheorem*{thm:repeat-hardness}{\Cref{cor:ma-hardness-setcsp} (restated)}
\begin{thm:repeat-hardness}
There exists some inverse polynomial $p(x) = \Theta(1/x^3)$ such that for every inverse polynomial $p' < p$, the problem $\SetCSP_{p'(m)}$ is $\MA$-hard.
\end{thm:repeat-hardness}
\begin{proof}
  It follows directly from \Cref{lem:hardness-completeness,}, together with the fact that, regarding the parameters in \Cref{lem:hardness-soundness}, we have that $T,q \leq m$. 
\end{proof}

\section{Reduction from $\SetCSP$ to $\ACleanCC$}
\label{sec:reduction}
In this section we reduce the $\SetCSP$ problem to the $\ACleanCC$, showing the containment of $\SetCSP$ in $\MA$ and the $\MA$-hardness of $\ACleanCC$.

Before showing the reduction, we prove a technical lemma that shows how, for a fixed $S$, the value of $\setunsat(\mathcal{C},S)$ and the number of bad and longing strings for $\mathcal{C}$ are related.

\begin{lemma}\label{lem:bound-setunsat}
For some fixed  non-empty $S \subset \01^n$, let $B_{\mathcal{C}} = \bigcup_{i}B_{C_i}$ and $L_{\mathcal{C}} = \bigcup_{i}L_{C_i}$, the union of bad and longing strings for all set-constraints in $\mathcal{C}$, respectively. We have that $\setunsat(\mathcal{C},S) \leq \frac{1}{|S|}( |B_{\mathcal{C}}| + |L_{\mathcal{C}}|$).
\end{lemma}
\begin{proof}
By definition of $\setunsat(\mathcal{C},S)$ and $\setunsat(C_i,S)$, we have that
\begin{align*}
\setunsat(\mathcal{C},S) &= \frac{1}{m}\sum_{i = 1}^m \setunsat(C_i,S) =
\frac{1}{m |S|}\sum_{i = 1}^m |B_{C_i}| + |L_{C_i}| \leq 
\frac{1}{m|S|}\sum_{i = 1}^m |B_{\mathcal{C}}| + |L_{\mathcal{C}}| \\
&= \frac{1}{|S|}(|B_{\mathcal{C}}| + |L_{\mathcal{C}}|),
\end{align*}
where in the inequality we use the fact that $B_{C_i} \subseteq B_{\mathcal{C}}$ and $L_{C_i} \subseteq L_{\mathcal{C}}$.
\end{proof}

For some inverse polynomial $\eps$, we consider an instance $\calC$ of
$\kSetCSP_{\eps}$. From $\calC$, we construct the graph $G_{\calC} = (\01^n,
E)$, where $(x,y) \in E$ if there exists a set-constraint $C \in
\calC$ such that $x$ and $y$ are $C$-neighbors. We can define $C_{G_{\calC}}$ that on input $x \in \01^n$, outputs all neighbors of $x$ by inspecting all set-constraints of $\calC$. Finally, we define $C_M$ as the circuit that on input $x \in \01^n$, outputs if $x$ is a bad string for $\calC$, again by inspecting all of its set-constraints.

\begin{lemma}[Reduction from $\SetCSP$ to $\ACleanCC$]
  \label{lem:completeness-ma}
  For every $\eps$ we have that:
  \begin{itemize}
      \item If $\mathcal{C} = (C_1,...,C_m)$ is a yes-instance of $\kSetCSP_{\eps}$, then $(C_{G_{\calC}},C_M)$ is a yes-instance of $\ACleanCC_{\eps/2}$. 
      \item If 
$\mathcal{C} = (C_1,...,C_m)$ is a no-instance of $\kSetCSP_{\eps}$, then $(C_{G_{\calC}},C_M)$ is a no-instance of $\ACleanCC_{\eps/2}$.
  \end{itemize}
\end{lemma}
\begin{proof}
To prove the first part,   we show that a non-empty $S \subseteq \01^n$ such that $\setunsat(S,\calC)
  =0$ implies that the connected component of any string $x\in S$ in $G_{\calC}$ contains only good strings.
  To show this, we notice that $S$ is a union of connected components of $G_{\calC}$. In this case, 
  any of these connected components imply that $(C_{G_{\calC}},C_M)$ is a yes-instance of $\ACleanCC$.
  
  Suppose towards contradiction that there exists a string $x$ in $S$ which is connected 
  to a string $y$ outside of $S$ via an edge in $G_{\calC}$; that means that $x$ and $y$ are $C$-neighbors for some set-constraint $C$. But this means that $x$ 
  is a $C$-longing string for $S$, and this contradicts 
  $\setunsat(S,\calC)
  =0$. We finish this part of the proof by stressing that, by assumption, no elements in $S$ are marked, otherwise there would be a bad string in it.
  
  \medskip
  
  For the second part, we show that if there is a set of vertices $S$ such that $\partial(S,\overline{S}) < \eps|S|/2$ on $G_{\mathcal{C}}$ and the number of marked elements in $S$ is strictly less than $\eps|S|/2$, then $\setunsat(S,\calC) < \eps$. In this case, if $\mathcal{C}$ is a no-instance of $\kSetCSP_{\eps}$, then $(C_{G_{\mathcal{C}}},C_M)$ must be a no-instance of $\ACleanCC_{\eps/2}$.
  
  Notice that if there are at most 
   $\eps|S|/2$ edges between $S$ and $\overline{S}$, then there are at most $\eps|S|/2$ vertices in $S$ that are connected to $\overline{S}$ and, by definition of the edges in $G_{\mathcal{C}}$, we have that $S$ has at most $\eps|S|/2$ $\mathcal{C}$-longing strings.
We also have that the number of bad strings in $S$ is, by definition, the number of marked elements which is also strictly less than $\eps|S|/2$. Therefore, by \Cref{lem:bound-setunsat}, we have that
$\setunsat(\calC ,S)  < \eps$.
\end{proof}

We can now finally prove \Cref{thm:main}: 

\newtheorem*{thm:repeat-main}{\Cref{thm:main} (restated)}
\begin{thm:repeat-main}
There exists some inverse polynomial $p(x) = \Theta(1/x^3)$ such that for every inverse polynomial $p' < p$, the problems $\SetCSP_{p'(m)}$ and $\ACleanCC_{p'(m)}$ are $\MA$-complete, where $m$ is the size of the $\SetCSP$ or $\ACleanCC$ instance.
\end{thm:repeat-main}
\begin{proof}
From \Cref{cor:ma-hardness-setcsp} we have that for some 
$\tilde{p} = \Theta(1/x^3)$,
$\SetCSP_{\tilde{p}}$ is $\MA$-hard\footnote{Notice that we slightly abuse the notation here: in \Cref{cor:ma-hardness-setcsp}, we define the hardness in respect to the parameter $m$, the number of clauses; here, we call $m$ the size of the $\SetCSP$ instance, which is lower-bounded by its number of clauses.} and from \Cref{cor:containment} we have that $\ACleanCC_{\tilde{p}/2}$ is in $\MA$.

In \Cref{lem:completeness-ma}, we show a reduction $\SetCSP_{\tilde{p}}$ to $\ACleanCC_{\tilde{p}/2}$, which implies, together with the aforementioned results, that $\SetCSP_{\tilde{p}}$ is in $\MA$ and that  $\ACleanCC_{\tilde{p}/2}$
is $\MA$-hard. Therefore, we can pick $p(x) = \tilde{p}(x)/2$ and our statement holds.
\end{proof}

\section*{Acknowledgements} 
 We notice that the ACAC problem did not appear in the first version of the this work and we thank an anonymous reviewer who pushed us to define more natural problems.
We are grateful to Umesh Vazirani for asking the right question that led to this note. 
D.A. is grateful for the support of ISF grant 1721/17.
Most of this work was done while A.G. was affiliated to CWI and QuSoft, and part of it was done while A.G. was visiting the Simons Institute for the Theory of Computing.

\bibliographystyle{alpha}
\bibliography{stoquastic}

\appendix

\section{Complexity classes}\label{ap:classes} 

\defclass{\NP}{def:NP}{
A problem $A=(\ayes,\ano)$ is in \NP{} if and only if there exist two polynomials
  $p,q$ and a deterministic algorithm $D$, where $D$
  takes as input a string $x\in\Sigma^*$ and a $p(|x|)$-bit witness $y$, runs in time $q(|x|)$ and outputs a bit $1$ (accept) or $0$ (reject)
  such that:
}
{If $x\in\ayes$, then there exists a witness $y$
    such that $D$ outputs accept on $(x,y)$.}
{If $x\in\ano$, then for any witness $y$, $D$ outputs 
    reject on $(x,y)$.}

We can then generalize this notion, by giving the verification algorithm the power to flip random coins, leading to the complexity class \MA{}.

\defclass{\MA}{def:MA}{
A (promise) problem $A=(\ayes,\ano)$ is in \MA{} if and only if there exist two polynomials
  $p,q$ and 
 a probabilistic algorithm $R$, where $R$
takes as input a string $x\in\Sigma^*$ and a $p(|x|)$-bit witness $y$, runs in time $q(|x|)$ and decides
on acceptance or rejection of $x$
such that:  
}
{If $x\in\ayes$, there exists a witness $y$
    such that $R$ accepts $(x,y)$ with probability $1$.}
{If $x\in\ano$, for any witness $y$, $R$
    accepts $(x,y)$ with probability at most $\frac{1}{3}$.}

\section{Reversible Circuits} 
\label{ap:rev}

In classical complexity theory, we usually describe the verification procedure for
NP and \MA{} as algorithms. Such polynomial time algorithms can be converted into a uniform family of
polynomial-size Boolean circuits, which are usually described using 
$\AND$, $\OR$ and $\NOT$
gates~\cite{AroraBarak09}. As we will see later, these circuits are not suitable in our context
because we need our circuits to be {\em reversible}, namely made of reversible gates. By this we mean that the output string of the gate is a one-to-one function of the input string, and thus an inverse gate exists. 
Clearly, we can also invert the entire circuit made of such gates by running 
all inverses of the gates in reverse order. 

Boolean circuits consisting of the {\em universal} gateset
 $\{\NOT, \CNOT, 
\CCNOT\}$,\footnote{$\NOT$ is the $1$-bit gate that on input $a$ outputs $1\oplus a$; $\CNOT$ is the $2$-bit gate that on input $(a,b)$ outputs $(a,b \oplus a)$; finally
$\CCNOT$  is the $3$-bit gate that on input $(a,b,c)$ outputs $(a,b,c\oplus
ac)$, i.e., it flips the last bit iff $a = b = 1$. Notice that 
$\{\NOT,\CCNOT\}$ is already universal for classical computation and we just add $\CNOT$ for convenience.}
 can be made reversible, by making use of additional auxiliary bits
initialized to $0$. It is folklore that the overhead for this conversion is linear in the number of gates.  
We will assume all verifiers considered are 
reversible Boolean circuits. 

For randomized circuits, we can also assume the circuit uses only reversible gates, in the following way: we assume
that the random bits are provided as part of the initial string on which the circuit acts; we can view it as part of the input. The rest of the circuit is reversible. 

For decision problems, we use the convention that the first bit is used as the output of the circuit, namely, \textsf{accept} or \textsf{reject}.

\section{Graph theory and random walks}
\label{ap:graph}
For some undirected graph $G = (V,E)$, and $v \in V$, let $d_G(v) = |\{u \in V : \{v,u\}
\in E \}|$ be the degree of $v$. 
We may define a weight function for the edges $w : E \rightarrow \mathbb{R}^+$,
and if it is not explicitly defined, we assume that $w(\{u,v\}) = 1$ for all
$\{u,v\} \in E$. We slightly abuse the notation and for some $E' \subseteq E$, we denote
$w(E') = \sum_{e \in E'} w(e)$.

For some set of vertices $S \subseteq V$, we
define the boundary of $S$ as
$\partial_G(S) = \{\{u,v\} \in E: u \in S, v \not\in S\}$, the neighbors of $S$
as $N_G(S) = \{u \in \overline{S} : v \in S, \{u,v\} \in E\}$ 
(where $N_G(u)$ for some $u\in V$ is shorthand for $N_G(\{u\})$); and finally 
the volume of $S$
$vol_G(S) = \sum_{u \in S, v \in N_G(u)} w(\{u,v\})$ and the
conductance of $S$, $\phi_G(S) = \frac{w(\partial_G(S))}{vol_G(S)}$. 

We define notions related to random walks on weighted graphs; 
definitions 
are taken from \cite{GharanT12}.
One step in the random walk on the weighted graph $G$ starting on vertex $v$ is defined by moving to vertex $u$ with probability $p_{v,u} = \begin{cases}\frac{w(\{u,v\})}{vol_G(\{v\})}, & \text{if } \{u,v\}  \in E \\ 0 ,& \text{otherwise} \end{cases}$. One step in the {\em lazy} random walk on $G$ starting on $v$ is defined by staying in the same vertex with probability $\frac{1}{2}$, or with probability $\frac{1}{2}$, moving according to the above  random walk.
The
conductance of the weighted graph is defined as \[\phi(G) = \min_{\substack{S \subseteq V
\\
vol_G(S)
\leq \frac{vol_G(V)}{2}}}\phi_G(S).\]

The (unique) stationary distribution $\pi_G$ of the lazy random walk on a connected weighted graph $G$
is given by the following probability distribution on the vertices: $\pi_G(v) =
\frac{vol_G(\{v\})}{vol_G(V)}$.  The following lemma connects the conductance 
of a graph and its {\em mixing time}, i.e., the number of steps of a lazy random
walk on $G$ starting at an arbitrary vertex such that the distribution on the vertices is close to the stationary distribution.

\begin{lemma}[High conductance implies small mixing time~\cite{JerrumS89}]\label{lemma:mixing-time}
 \ Let $G = (V,E)$ be a undirected connected graph. For any $u,v \in V$, a
  $\left(\frac{2}{\phi(G)^2/4}\ln\frac{1}{\eps \min\{\pi_G(v)\}}\right)$-step lazy-random walk starting
  in $u$ ends in $v$ with probability $p^T_{u,v} \in [\pi_G(v)-\eps,\pi_G(v)+\eps]$.
 \footnote{
We use Corollary $2.2$ of \cite{JerrumS89}, along with the
  remark $(a)$ after it. The latter states that the conductance of a graph is divided by $2$ when we consider the weighted graph induced by the {\em lazy} random walk (this induced graph is naturally defined by adding a self loop to each vertex, and giving it the same weight as the sum of the weight of all other edges going out of the node). In \Cref{lemma:mixing-time}, we use the conductance of the original graph $G$ and account for the fact that we are applying a lazy random walk on it by dividing the conductance by $2$. Importantly, the stationary distribution remains the same.}
\end{lemma}

We now prove  \Cref{lem:hitting-time}
\newtheorem*{thm:repeat-hitting}{\Cref{lem:hitting-time} (restated)}
\begin{thm:repeat-hitting}[Escaping time of high conductance subset]
  Let $G = (V =A\cup B,E)$ be a simple (no multiple edges) 
  undirected connected graph such that for
  every $v \in A$
  $d_G(v) \leq d$ , and such that for some $\delta < \frac{1}{2}$, 
  for all $A' \subseteq A$, we have that 
  $|\partial_G(A')| \geq \delta |A'|$.
Then  a $\left( 
   \stepsrandomwalk
  \right)$-step lazy random walk starting in any $v \in A$ reaches some vertex $u
  \in B$ with probability at least $\frac{\delta}{4d}$.
\end{thm:repeat-hitting}
\begin{proof}
  Let us consider the graph $G' = (V',E')$ where $V' = A \cup \{b^*\}$ and $E' =
  \{ \{u,v\} \in E: u,v\in A\} \cup \{\{u,b^*\} : u \in A, v \in B, \{u,v\} \in
  E\}$ (as a set, duplicate edges are removed).  In other words, $G'$ can be seen as the graph derived when we contract all
  vertices in $B$ to one vertex (and call this vertex $b^*$) and contract the parallel edges that appear when we perform the contraction of the vertices
  into a single edge. For some $u \in A$, we define the 
  weight of the edge $\{u,v\} \in E'$ as 
  \[w(\{u,v\}) = \begin{cases} 1, &
  \text{if } v \in A \\
  \left|\left\{\{u,v'\} : \{u,v'\}\in E, v' \in B\right\}\right|,&\text{if } v =
  b^*\end{cases}.\]
  
  Notice that the probability that a $T$-step weighted lazy random walk on $G'$
  starting in $v \in A$ passes by the vertex $b^*$
  is exactly the same as the probability that a $T$-step lazy random walk on
  $G$ starting in $v$ passes by some vertex in $B$. Therefore, we continue by
  lower bounding the former, and the statement will follow.

  Let us start by analyzing $\phi(G')$. Notice that for all $A' \subseteq
  A$, it follows that
  \begin{align}\label{eq:conductance-1}
    \phi_{G'}(A') = \frac{|\partial_{G'}(A')|}{vol_{G'}(A')} =   \frac{|\partial_G(A')|}{\sum_{v\in A'}2 d_G(v)}\geq \frac{\delta|A'|}{d|A'|}=\frac{\delta}{d}.
  \end{align}

  Let $S = \{b^*\} \cup A'$ for some $A' \subseteq A$ such that $vol(S) \leq
  \frac{vol_{G'}(V')}{2}$, we have that
  \begin{align}\label{eq:conductance-2}
    \phi_{G'}(S) = \frac{w(\partial_{G'}(S))}{ vol_{G'}(S)}
    \geq \frac{w(\partial_{G'}(\overline{S}))}{ vol_{G'}(\overline{S})}=\frac{|\partial_{G}(\overline{S})|}{\sum_{v\in \overline{S}} d_G(v)}
  \geq \frac{|\partial_{G}(\overline{S})|}{ d|\overline{S}|}
    \geq \frac{\delta |\overline{S}|}{ d|\overline{S}|}
    = \frac{\delta}{d},
  \end{align}
  where in the first inequality we use the fact that 
  $\partial_{G'}(S) = \partial_{G'}(\overline{S})$ and $vol_{G'}(S) \leq
  vol_{G'}(\overline{S})$, in the second equality we use the definition of $w$, in the second inequality we use the fact that $\overline{S} \subseteq A$ and every vertex in $A$ has degree at most $d$, 
  and the third inequality follows since $\overline{S} \subseteq A$, 
  and so we can use the assumption of the lemma. 
  From \Cref{eq:conductance-1,eq:conductance-2}, we
  have that $\phi(G') \geq \frac{\delta}{d}$.

We now analyze the stationary distribution of the lazy random walk associated with $G'$. Since $\delta|A|\leq vol_{G'}(\{b^*\}) \leq d|A|$ and for any $v\in A$, $1\leq vol_{G'}(\{v\}) \leq d$, we have that $vol_{G'}(V) \leq (\delta + d)|A| \leq 2d|A|$,
  \[
  \pi_{G'}(b^*)= \frac{vol_{G'}(\{b^*\})}{vol_{G'}(V)} \geq \frac{\delta |A|}{2d|A|} = \frac{\delta}{2d},\] and for any $v \in A$,
  \[\pi_{G'}(v)  = \frac{vol_{G'}(\{v\})}{vol_{G'}(V)} \geq
  \frac{1}{2d|A|} \geq \frac{1}{2d|V|}.\]
  We now apply~\Cref{lemma:mixing-time},
  using
  $\eps=\frac{\delta^2}{2d}$, $\min\{\pi_{G'}(v)\} \geq \frac{1}{2d|V|}$ and  $\phi(G') \geq \frac{\delta}{d}$.
  We have
  that for 
  \[T =
  \frac{8}{\phi(G')^2}\ln\frac{1}{\eps \min\{\pi_{G'}(v)\}} \leq
  \frac{8 \cdot d^2}{\delta^2}\ln\frac{4d^2|V|}{\delta^2}  \leq \stepsrandomwalk,\]

  a $T$-step lazy random walk starting from any vertex
  end in $b^*$
  with probability at least $\frac{\delta- \delta^2}{2d}$,
  which is a lower bound on the probability that the lazy random walk passes by
  $b^*$. Since $\frac{\delta- \delta^2}{2d}\geq\frac{\delta}{4d}$ this finishes the proof. 
\end{proof}

\subsection{Hypercube}
Given some parameter $q$, the $q$-dimensional hypercube is the graph where the vertices are all possible $q$-bit strings and two vertices are connected if the corresponding strings differ on exactly one bit.

We use the following well-known fact about the expansion of the hypercube.
\begin{lemma}[\cite{RV11}]\label{lem:expansion-hypercube}
Let $G$ be the $q$-dimensional hypercube. Then $\phi(G) = \frac{1}{q}$. 
\end{lemma}

\section{\class{PSPACE}-completeness of \CleanCC{}}
\label{app:pspace}
In this section we give a sketch of proof that the \CleanCC{} problem is \class{PSPACE}-complete. 

\begin{lemma}
\CleanCC{} is in \PSPACE{}.
\end{lemma}
\begin{proof}[Proof sketch]
In order to prove containment, we can follow two steps
\begin{enumerate}
    \item A variant of \CleanCC{}, called \FixedCleanCC{}, is in \PSPACE{}.
    \item The containment of \FixedCleanCC{} in \PSPACE{} implies the containment of \CleanCC{} in \PSPACE{}.
\end{enumerate}

The problem \FixedCleanCC{} consists of deciding if the connected component of a fixed node contains only unmarked elements.
More concretely, it corresponds to the language  \[\{(x,C_G,C_M) : \text{the connected component of } x \text{ in the graph } G \text{ contains only unmarked elements} \},\] where $C_G$ and $C_M$ are defined as in \CleanCC{}. 
We argue that \FixedCleanCC{} is in \PPSPACE{},\footnote{\PPSPACE{} is the complexity class of randomized algorithms that run in polynomial space and non-zero completeness/soundness gap} and since \PPSPACE{}=\PSPACE{} \cite{Ladner89}, this proves the first claim. The \PPSPACE{} algorithm consists of a $2^{n}$-step random walk on $G$ starting from $x$ which rejects iff a marked element is encountered. If the connected component of $x$ has no marked elements, so the algorithm never rejects, whereas for no-instances, the random walk finds the marked element with non-zero probability.

 We can devise a polynomial-space algorithm for $\CleanCC$ that iterates over all $x \in \01^n$ and accept iff for at least one of such $x$'s, $(x,C_G,C_M)$ is a yes-instance of $\FixedCleanCC$ (which can be decided in $\PSPACE$ as described before).
 \end{proof}

\begin{lemma}
\CleanCC{} is \PSPACE{}-hard.
\end{lemma}
\begin{proof}[Proof sketch.]
Let $L$ be a language in \class{PSPACE}. Since we can assume that any problem in \class{PSPACE} can also be solved in polynomial space by a {\em reversible} Turing machine~\cite{LangeMT00}, we fix the reversible Turing Machine $M$ that decides if $x$ belongs to $L$ in $p(|x|)$-space, for some polynomial $p$. A {\em configuration} of $M$ consists of its internal state, the content of its tape and the position of the tape-head. For some configuration $w$, we denote $pred(w)$ as the preceding configuration of $w$ and $succ(w)$ as its following configuration (and both can be efficiently computed from $w$ since $M$ is reversible).

In the reduction to \CleanCC{}, we set $n = p(|x|) + S$, where $S = O(p(|x|))$ is the amount of memory to store a snapshot of the computation of $M$ (meaning the current state of the Turing machine, the tape and the position of the head). Each string $y \in \01^n$ can be split in two parts: the first one is the ``clock''-substring and the second one is the 
``working''-substring. The clock substring represents, in binary, a value in $t \in [2^{p(|x|)}]$ (notice that $M$ runs in time at most $2^{p(n)}$) and the working substring represents an arbitrary state of the Turing Machine at time $t$. 

We define then instance $(C_G,C_M)$ of \CleanCC{} as follows:
\begin{itemize}
    \item On $t \concat w \in \01^n$, $C_G$ outputs the string $t-1 \concat pred(w)$ and $t+1 \concat succ(w)$ (if $t = 0$ or $t = 2^n$, $C_G$ only outputs the meaningful neighbors).
    \item On $t \concat w \in \01^n$, $C_M$ outputs $1$ if $t = 0$ and $w$ does not consist of the state of $M$ at the beginning of the computation or if $t = 2^{p(|x|)}$ and $w$ does not consist of the state of $M$ that would lead to acceptance.  
\end{itemize}

It is not hard to see that if $x \in L$, then the connected component of $0 \concat x $ is clean, whereas if $x \not\in L$, all connected components have at least one marked element.
\end{proof}
\end{document}